\documentclass[12pt,notitlepage]{article}
\usepackage{eurosym}
\usepackage{amssymb}
\usepackage{amsmath}
\usepackage{amsmath}
\usepackage{amsthm}
\usepackage{amsfonts}
\usepackage{graphicx}
\usepackage{enumerate}
\usepackage{dsfont}
\usepackage{amssymb,subfig,graphicx}
\usepackage{fancyhdr}
\usepackage{placeins}
\usepackage{setspace}
\usepackage{todonotes}
\usepackage{pdfsync}
\usepackage[numbers]{natbib}

\numberwithin{equation}{section}
\newtheorem{theorem}{Theorem}[section]

\newtheorem{condition}[theorem]{Condition}

\newtheorem{corollary}{Corollary}[section]

\newtheorem{example}{Example}

\newtheorem{remark}{Remark}[section]

\newtheorem{ass}{Assumption}[section]
\newtheorem{lemma}{Lemma}

\theoremstyle{definition}

\DeclareMathOperator{\E}{\mathbb{E}}

\def\independenT#1#2{\mathrel{\rlap{$#1#2$}\mkern2mu{#1#2}}}

\newcommand{\supp}{\mathrm{supp}\,}

\newcommand{\R}{\mathbb{R}}

\begin{document}

\title{Nonparametric Identification of Random Coefficients in Endogenous and
Heterogeneous Aggregate Demand Models}
 \author{%
 \begin{tabular}{ccccc}
 Fabian Dunker\thanks{%
 School of Mathematics and Statistics, University of Canterbury, Private
 Bag 4800, Christchurch 8140, New Zealand, Email: fabian.dunker@canterbury.ac.nz} &  & Stefan Hoderlein\thanks{%
 Department of Economics, Emory University, 1602 Fishburn Dr, Atlanta, GA 30322, USA, Email: stefan.hoderlein@emory.edu. } &  & Hiroaki Kaido%
 \thanks{%
 Department of Economics, Boston University, 270 Bay State Road, Boston, MA
 02215, USA, Email: hkaido@bu.edu. We thank Steve Berry, Jeremy Fox, Amit
 Gandhi, Phil Haile, Kei Hirano, Arthur Lewbel, Marc Rysman, and Elie Tamer
 for their helpful comments. We also thank seminar and conference
 participants at Harvard, BU-BC joint mini-conference, the Demand Estimation
 and Modeling Conference 2013 and the Cowles Summer Conference 2014. Kaido gratefully acknowledges financial
 support from NSF Grant SES-1357643.} \\
 {\small University of Canterbury } &  & {\small Emory University } &  &
 {\small Boston University} \\
 &  &  &  &  \\
 &  &  &  &
 \end{tabular}%
 }
 \date{\today }
\maketitle

\begin{abstract}
This paper studies nonparametric identification in market level demand
models for differentiated products with heterogeneous consumers. We consider
a general class of models that allows for the individual specific
coefficients to vary continuously across the population and give conditions
under which the density of these coefficients, and hence also functionals
such as welfare measures, is identified. 
A key finding is that two leading models, the BLP-model
(Berry, Levinsohn, and Pakes, 1995) and the pure characteristics model
(Berry and Pakes, 2007), require considerably different conditions on the
support of the product characteristics.

\end{abstract}

%\clearpage

\section{Introduction}

Modeling consumer demand for products that are bought in single or discrete
units has a long and colorful history in applied economics, dating back to
at least the foundational work of McFadden (1974, 1981). While allowing for
heterogeneity, much of the earlier work on this topic, however, was not able
to deal with the fact that in particular the own price is endogenous. In a
seminal paper that provides the foundation for much of contemporaneous work
on discrete choice consumer demand, Berry, Levinsohn and Pakes (1994, BLP)
have proposed a solution to the endogeneity problem. Indeed, this work is so
appealing that it is not just applied in discrete choice demand and
empirical IO, but also increasingly in many adjacent fields, such as health,
urban or education economics, and many others. From a methodological
perspective, this line of work is quite different from traditional
multivariate choice, as it uses data on the aggregate level and integrates
out individual characteristics\footnote{%
There are extensions of the BLP framework that allow for the use of
Microdata, see Berry, Levinsohn and Pakes (2004, MicroBLP). In this paper,
we focus on the aggregate demand version of BLP, and leave an analogous work
to MicroBLP\ for future research.} to obtain a system of nonseparable
equations. This system is then inverted for unobservables for which in turn
a moment condition is then supposed to hold.

Descending in parts from the parametric work of McFadden (1974, 1981),
market-level demand models share many of its features, in particular
(parametric) distributional assumptions, but also a linear random
coefficients (RCs) structure for the latent utility. Not surprisingly, there
is increasing interest in the properties of the model, in particular which
features of the model are nonparametrically point identified, and how the
structural assumptions affect identification of the parameters of interest.
Why is the answer to these questions important? Because an empiricist
working with this model wants to understand whether the results she obtained
are a consequence of the specific parametric assumptions she invoked, or
whether they are at least qualitatively robust. In addition, nonparametric
identification provides some guidance on essential model structure and on
data requirements, in particular about instruments. Finally, understanding
the basic structure of the model makes it easier to understand how the model
can be extended. Extensions of the BLP framework that are desirable are in
particular to allow for consumption of bundles and multiple units of a product  without modeling every choice as a new separate alternative.

We are not the first to ask the nonparametric identification question for
market demand models. In a series of elegant papers, Berry and Haile (2014, BH henceforth), Berry and Haile (2020) provide important answers to many of the identification
questions. In particular, they establish conditions under which the
\textquotedblleft Berry inversion\textquotedblright , a core building block
of the BLP\ model named after Berry (1994), which allows to solve for
unobserved product characteristics, as well as the distribution of a
heterogeneous utility index are nonparametrically identified.

Our work complements this line of work in that we follow more closely the
original BLP specification and assume in addition that the utility index has
a linear random coefficients (RCs) structure. More specifically, we show how
to nonparametrically identify the distribution of random coefficients in
this framework. This result does not just close the remaining gap in the
proof of nonparametric identification of the original BLP model, but is also
important for applications because the distribution of random coefficients
allows to characterize the distribution of the changes in welfare due to a
change in observable characteristics, in particular the own price (to borrow an analogy from
the treatment effect literature, if we think of a price as a treatment, BH
recover the treatment effect on the distribution, while we recover the
distribution of treatment effects). For example, consider a change in the
characteristics of a good. The change may be due to a new regulation, an
improvement of the quality of a product, or an introduction of a new
product. Knowledge of the random coefficient density allows the researcher
to calculate the distribution of the welfare effects. This allows one to
answer various questions. For example, one may investigate whether the
change gives rise to a Pareto improvement. This is possible because, with
the distribution of the random coefficients being identified, one can track
each individual's welfare before and after the change. If a change in one of
the product characteristics is not Pareto improving, one can also calculate
the proportion of individuals who would benefit from the change and
therefore prefers the product with new characteristics.\footnote{%
Note that simultaneous changes in product characteristics and price are
allowed. Hence, one can investigate how much price change is required
to compensate for a change (e.g. downgrading of a feature) in one of the
product characteristics to let a certain fraction of individuals receive a
non-negative utility change, i.e. $P(\Delta U_{ijt}\geq 0)\geq \tau $ for
some prespecified $\tau \in \lbrack 0,1]$, where $\Delta U_{ijt}$ denotes
the utility change.} Identification of the random coefficient distribution
allows one to conduct various types of welfare analysis that are not
possible by only identifying the demand function. Our focus therefore will
be on the set of conditions under which one can uniquely identify the random
coefficient distribution from the observed demand.

Naturally, identification will depend crucially on the specific model at
hand. As it turns out, there are important differences between the classical
BLP and the pure characteristics model (see Berry and Pakes (2007), PCM\
henceforth) that stem from the presence of an alternative, individual and
market specific error, typically assumed to be logistically distributed and
hence called \textquotedblleft logit error\textquotedblright\ in the
following. A lucid discussion about the pros and cons of both approaches can
be found in Berry and Pakes (2007). One advantage of the PCM we would like
to emphasize at this point is that it is well-suited for the analysis of
welfare changes when a new product with a particular characteristic is
introduced to the market. Moreover, the pure characteristics model also
predicts a reasonable substitution pattern when the number of products is
large, while the BLP-type model may give counter-intuitive predictions. In
addition to these important economic differences, the identification
strategies including the required assumptions also differ significantly
across the two models. In particular, in the BLP model, one needs to rely on
an identification at infinity argument to isolate the unobservable for each
product. In remarkable contrast, in the PCM\ one does not require such an
argument (and therefore does not have to employ some restrictive
assumptions). Instead, in the PCM, we demonstrate that one may combine
demand on products across different markets to construct a function that
depends on the random coefficients through a single index so that we can
recover the distribution of unobserved heterogeneity without relying on
identification at infinity. We call this construction \emph{marginalization
(or aggregation) of demand}. This is possible due to the unique structure of
the PCM in which  only the product characteristics (but not the tastes for
products) determine the demand. To our knowledge, this identification 
strategy is novel.

The arguments in establishing nonparametric identification of these changes
are constructive and permit the construction of sample counterparts
estimators, using the theory in Hoderlein, Klemel\"{a} and Mammen (2010) or 
Dunker Mendoza and Reale (2021). This theory reveals that the random
coefficients density is only weakly identified, suggesting that numerical
instabilities and problems frequently reported and discussed in the BLP
literature, e.g., Dube, Fox and Su (2013), are caused or aggravated by this
feature of the model.

Another contribution in this paper is that we use the insights obtained from
the identification results to extend the market demand\ framework to cover
bundle choice (i.e., consume complementary goods together). Note that bundles can
in principle be accommodated within the BLP framework by treating them as
separate alternatives. However, this is not parsimonious as the number of
alternatives increases rapidly and with it the number of unobserved product
characteristics, making the system quickly intractable. To fix ideas,
suppose there were two goods, say good A and B. First, we allow for the
joint consumption of goods A and B, and second, we allow for the consumption
of several units of either A and/or B, without labeling it a separate
alternative. We model the utility of each bundle as a combination of the
utilities for each good and an extra utility from consuming the bundle. This
structure in turn implies that the dimension of the unobservable product
characteristic equals the number of goods $J$ instead of the number of
bundles. There are three conclusions we draw from this contribution:\ first,
depending on the type of model, the data requirements vary. In particular,
to identify all structural parts of the model, in, say, the model on bundle
choice, market shares are not the correct dependent variable any more.
Second, depending on the object of interest, the data requirements and
assumptions may vary depending on whether we want to just recover demand
elasticities, or the entire distribution of random coefficients. Third, the
parsimonious features of the structural model result in significant
overidentification of the model, which opens up the way for specification
testing, and efficient estimation. As in the classical BLP setup, in all
setups we may use the identification argument to propose a nonparametric
sample counterpart estimators.\footnote{In DHK (2018), we also use the insights obtained to
propose a parametric estimator for models where there had not been an
estimator before.}

\textbf{Related literature: }as discussed above, this paper is closely
related to both the original BLP line of work (Berry, Levinsohn and Pakes
(1994, 2004)), as well as to the recent identification analysis of Berry and
Haile (2014, 2020). Because of its generality, our approach also provides
identification analysis for the \textquotedblleft pure
characteristics\textquotedblright\ model of Berry and Pakes (2007), see also
Ackerberg, Benkard, Berry and Pakes (2007) for an overview. Other important
work in this literature that is completely or partially covered by the
identification results in this paper include Petrin (2002) and Nevo (2001).
Moreover, from a methodological perspective, we note that BLP continues a
line of work that emanates from a broader literature which in turn was
pioneered by McFadden (1974, 1981); some of our identification results
extend therefore beyond the specific market demand model at hand. Other
important recent contributions in discrete choice demand include
Gowrisankaran and Rysman (2012), Armstrong (2016) and Moon, Shum, and
Weidner (2018). Less closely related is the literature on hedonic models,
see Heckman, Matzkin and Nesheim (2010), and references therein.

In addition to this line of work, we also share some commonalities with the
work on bundle choice in IO, most notably Gentzkow (2007), and Fox and
Lazzati (2017). For some of the examples discussed in this paper, we use
Gale-Nikaido inversion results, which are related to arguments in Berry,
Gandhi and Haile (2013). Because of the endogeneity, our approach also
relates to nonparametric IV, in particular to Newey and Powell (2003),
Andrews (2017), and Dunker, Florens, Hohage, Johannes, and Mammen (2014).
Finally, our arguments are related to the literature on random coefficients
in discrete choice model, see Ichimura and Thompson (1995), Gautier and
Kitamura (2013), Fox and Gandhi (2016), Dunker, Hoderlein, Kaido, and Sherman (2018),
and Matzkin (2012). Since we use the Radon transform introduced by
Hoderlein, Klemel\"{a} and Mammen (2010, HKM) into Econometrics, 
this work is particularly close to the literature that uses the Radon
transform, in particular HKM\ and Gautier and Hoderlein (2015). Finally, the
class of models we consider is related but differs from the mixed logit
model (without endogeneity) analyzed by Fox, Kim, Ryan, and Bajari (2012)
who established the identification of the distribution of the random
coefficients from micro-level data, while maintaining the logit assumption
on the tastes for products. Our focus here is on market-level models with
endogeneity with the main goal being the identification of the distribution
of all random coefficients without any parametric assumption. As such, our
identification strategy differs significantly from theirs. Finally, after the original version of this paper, there have been recent developments on the nonparametric identification of aggregate demand models. Allen and Rehbeck (2020) study partial identification of latent complementarity in an aggregate demand model of bundles. Their focus is on what can be learned about latent complementarity when the variation of demand shifters is limited. Lu, Shi, and Tao (2019) study identification and semiparametric estimation of random coefficient logit demand models in a related but different environment, in which consumers face a growing number of products.

\textbf{Structure of the paper}: The second section lays out preliminaries
we require for our main result: We first introduce the class of models and
detail the structure of our two main setups. Still in the same section, for
completeness we quickly recapitulate the results of Berry and Haile (2014)
concerning the identification of structural demands, adapted to our setup.
The third section contains the key novel result in this paper, the
nonparametric (point-)identification of the distribution of random
coefficients in the class of discrete choice demand model with
endogeneity, which includes the BLP and PCM models. In the fourth section,
we discuss the identification in the bundles
case, including how the structural demand identification results of Berry
and Haile (2014) have to be adapted, but again focusing on the random
coefficients density. We then end with an outlook.

\section{Preliminaries}

\subsection{Model}

We begin with a setting where a consumer faces $J\in\mathbb{N}$ products and
an outside good which is labeled good 0. Throughout, we index individuals by 
$i$, products by $j$ and markets by $t$. We use upper-case letters, e.g. $%
X_{jt}$, for random variables (or vectors) that vary across markets and
lower-case letters, e.g. $x_{j}$, for particular values the random variables
(vectors) can take. In addition, we use letters without a subscript for
products e.g. $X_{t}$ to represent vectors e.g. $(X_{1t},\cdots ,X_{Jt})$.
For individual $i$ in market $t$, the (indirect) utility from consuming good 
$j$ depends on its (log) price $P_{jt}$, a vector of observable
characteristics $X_{jt}\in \mathbb{R}^{d_X}$, and an unobservable scalar
characteristic $\Xi _{jt}\in \mathbb{R}$. We model the utility from
consuming good $j$ using the linear random coefficient specification: 
\begin{equation}
U_{ijt}^{\ast }\equiv X_{jt}^{\prime }\beta _{it}+\alpha _{it}P_{jt}+\Xi
_{jt}+\sigma_\epsilon\epsilon_{ijt},~j=1,\cdots ,J~,  \label{eq:randomu}
\end{equation}%
where $(\alpha _{it},\beta _{it})^{\prime }\in \mathbb{R}^{d_X+1}$ is a
vector of random coefficients representing the tastes for the product
characteristics. For each $j$, $\epsilon_{ijt}$ represents the ``taste for
the product'' itself. Following Berry and Pakes (2007), we consider a class
of general market-level demand models that nests models with tastes for products 
($\sigma_\epsilon=1$) and without tastes for products ($\sigma_\epsilon=0$). 
The models with tastes for products include the random coefficient logit model 
used in BLP, in which case $\sigma_\epsilon=1$ and $\epsilon_{ijt},j=1,\cdots,J$ 
are i.i.d. Type-I
extreme value random variables. When $\sigma_\epsilon=0$, the model is
called the \textit{pure characteristic model} (PCM). The two models are
known to have different theoretical properties. For example, the BLP model
predicts that even with a large number of products, the mark-up remains
positive implying there is always an incentive to develop a new product. As
the number of new products grows, each individual's utility tends to
infinity.  On the other hand, in PCM, the model approaches competitive
equilibrium and the incentive to develop a new product diminishes as the
number of products increases.\footnote{%
See Berry and Pakes (2007) for more details.} As we will show below, the two
models also differ in terms of empirical contents.

Throughout, we assume that $X_{jt}$ is exogenous, while $P_{jt}$ can be
correlated with the unobserved product characteristic $\Xi _{jt}$ in an
arbitrary way. Without loss of generality, we normalize the utility from the
outside good to 0. This mirrors the setup considered in BH (2014).

We think of a large sample of individuals as $iid$ copies of this population
model. The random coefficients $\theta_{it}\equiv(\alpha_{it}
,\beta_{it},\epsilon_{i1t},\cdots,\epsilon_{iJt})^{\prime }$ vary across
individuals in any given market (or, alternatively, have a distribution in
any given market in the population), while the product characteristics vary
solely across markets. These coefficients are assumed to follow a
distribution with a density function $f_{\theta }$ with respect to Lebesgue
measure, i.e., be continuously distributed.\footnote{%
This assumption is not crucial but made for the ease of exposition. Our main
identification results (Theorems \ref{thm:fbeta} and \ref{thm:fbeta_pure}) 
hold for any Borel measure.} This density is assumed to
be common across markets, and is therefore not indexed by $t.$ As we will
show, an important aspect of our identification argument is that, once the
demand function is identified, one may recover $\Xi_{t}$ from the market
shares and other product characteristics $(X_{t},P_{t})$. Then, by creating
exogenous variations in the product characteristics and exploiting the
linear random coefficients structure, one may trace out the distribution $%
f_\theta$ of the preference that is common across markets. We note that we
can allow for the coefficients $(\alpha_{it} ,\beta_{it})$ to be alternative 
$j$ specific, and will do so in the online supplement. However, parts of the analysis
will subsequently change.%, and we start out with the more common case where the coefficients are the same across $j$.

Having specified the model on the individual level, the outcomes of individual
decisions are then aggregated in every market. The econometrician observes
exactly these market level outcomes $S_{l,t}$, where $l$ belongs to some
index set denoted by $\mathbb{L}$. Below, we give two examples. The first
example is the setting of the BLP and pure characteristics models, where
individuals choose a single good out of multiple products, while the second
is about the demand for bundles.

\begin{example}[Multinomial choice]
\textrm{\textrm{\textrm{Each individual chooses the product that maximizes
her utility out of $J\in\mathbb{N}$ products. Hence, product $j$ is chosen
if 
\begin{align}
U^*_{jt}>U^*_{kt}~,~~\forall k\ne j~.
\end{align}
The demand for good $j$ in market $t$ is obtained by aggregating the
individual demand with respect to the distribution of individual
preferences. 
\begin{align} \label{eq:blp1}
\begin{split}
\varphi_j(X_t,P_t,\Xi_t)&=\int 1\{X_{jt}^{\prime }b+aP_{jt}+\sigma_\epsilon
e_{j}>-\Xi_{jt}\}\\
&\times 1\{(X_{jt}-X_{1t})^{\prime
}b+a(P_{jt}-P_{1t})+\sigma_\epsilon(e_{j}-e_{1})>-(\Xi_{jt}-\Xi_{1t})\} \\
\cdots 1\{(X_{jt}-&X_{Jt})^{\prime
}b+a(P_{jt}-P_{Jt})+\sigma_\epsilon(e_j-e_J)>-(\Xi_{jt}-\Xi_{Jt})\}f_%
\theta(b,a,e)d\theta~, 
\end{split}
\end{align}
for $j=1,\cdots,J$, while the aggregate demand for good 0 is given by 
\begin{multline*}
\varphi_0(X_t,P_t,\Xi_t)=\int 1\{X_{1t}^{\prime }b+aP_{1t}+\sigma_\epsilon
e_1<-\Xi_{1t}\}\\[-6pt]
\cdots 1\{X_{Jt}^{\prime }b+aP_{Jt}+\sigma_\epsilon
e_J<-\Xi_{Jt}\}f_\theta(b,a,e)d\theta~,
\end{multline*}
where $(b,a,e_1,\cdots,e_J)$ are placeholders for the coefficients $%
\theta_{it}=(\beta_{it},\alpha_{it},\epsilon_{i1t},\cdots,\epsilon_{iJt})$.
The researcher then observes the market shares of products $%
S_{lt}=\varphi_l(X_t,P_t,\Xi_t), l\in \mathbb{L}$, where $\mathbb{L}%
=\{0,1,\cdots, J\}.$ } } }
\end{example}

The second of example considers discrete choice, but allows for the
choice of bundles.

\begin{example}[Bundles]\label{ex:bundles}\rm
Each individual faces $J=2$
products and decides whether or not to consume a single unit of each of the
products. There are therefore four possible combinations $(Y_1,Y_2)$ of
consumption units, which we call \emph{bundles}. In addition to the utility
from consuming each good as in \eqref{eq:randomu}, the individuals gain
additional utility (or disutility) $\Delta_{it}$ if the two goods are
consumed simultaneously. Here, $\Delta_{it}$ is also allowed to vary across
individuals. The utility $U^*_{i,(Y_1,Y_2),t}$ from each bundle is therefore
specified as: 
\begin{align}
&U^*_{i,(0,0),t}=0,  \notag \\
&U^*_{i,(1,0),t}=X_{1t}^{\prime }\beta_{it}+\alpha_{it}
P_{1t}+\Xi_{1t}+\sigma_\epsilon\epsilon_{i1t},~~
U^*_{i,(0,1),t}=X_{2t}^{\prime }\beta_{it}+\alpha_{it}
P_{2t}+\Xi_{2t}+\sigma_\epsilon\epsilon_{i2t},  \notag \\
&U^*_{i,(1,1),t}=X_{1t}^{\prime }\beta_{it}+X_{2t}^{\prime
}\beta_{it}+\alpha_{it} P_{1t}+\alpha_{it} P_{2t}
+\Xi_{1t}+\Xi_{2t}+\sigma_\epsilon\epsilon_{i1t}+\sigma_\epsilon%
\epsilon_{i2t}+\Delta_{it}~,  \label{eq:bundleutil}
\end{align}
Each individual chooses a bundle that maximizes her utility. Hence, bundle $%
(y_1,y_2)$ is chosen when $U^*_{i,(y_1,y_2),t}>U^*_{i,(y^{\prime
}_1,y^{\prime }_2),t}$ for all $(y^{\prime }_1,y^{\prime }_2)\ne (y_1,y_2)$.
For example, bundle $(1,0)$ is chosen if 
\begin{align}\label{eq:10demand}
\begin{split}
&X_{1t}^{\prime }\beta_{it}+\alpha_{it}
P_{1t}+\Xi_{1t}+\sigma_\epsilon\epsilon_{i1t}> 0,\\
&X_{1t}^{\prime }\beta_{it}+\alpha_{it}
P_{1t}+\Xi_{1t}+\sigma_\epsilon\epsilon_{i1t}>X_{2t}^{\prime
}\beta_{it}+\alpha_{it} P_{2t}+\Xi_{2t}+\sigma_\epsilon\epsilon_{i2t},\\
&X_{1t}^{\prime }\beta_{it}+\alpha_{it}
P_{1t}+\Xi_{1t}+\sigma_\epsilon\epsilon_{i1t}\\
&\hspace{0.4in}>X_{1t}^{\prime
}\beta_{it}+\alpha_{it}
P_{1t}+\Xi_{1t}+\sigma_\epsilon\epsilon_{i1t}+X_{2t}^{\prime
}\beta_{it}+\alpha_{it}
P_{2t}+\Xi_{2t}+\sigma_\epsilon\epsilon_{i2t}+\Delta_{it}~.
\end{split}
\end{align}
Suppose the random coefficients $\theta_{it}=(\beta_{it}^{\prime
},\alpha_{it},\Delta_{it},\epsilon_{i1t},\epsilon_{i2t})$ have a joint
density $f_\theta$. The aggregate structural demand for $(1,0)$ can then be
obtained by integrating over the set of individuals satisfying %
\eqref{eq:10demand} with respect to the distribution of the random
coefficients: 
\begin{align}\label{eq:10agdemand}
\begin{split}
\varphi_{(1,0)}(X_t,P_t,\Xi_t)=\int 1\{&X_{1t}^{\prime }b+a
P_{1t}+\sigma_\epsilon e_1>-\Xi_{1t}\} \\
\times 1\{(X_{1t}&-X_{2t})^{\prime
}b+a(P_{1t}-P_{2t})+\sigma_\epsilon(e_1-e_2)>\Xi_{2t}-\Xi_{1t}\} \\
&\times 1\{X_{2t}^{\prime }b+a P_{2t}+\sigma_\epsilon
e_2+\Delta<-\Xi_{2t}\}f_\theta(b,a,\Delta,e)d\theta~.  
\end{split}
\end{align}
The aggregate demand on other bundles can be obtained similarly. The
econometrician then observes a vector of aggregate demand on the bundles: $%
S_{l,t}=\varphi_{l}(X_t,P_t,\Xi_t),l\in\mathbb{L}$ where $\mathbb{L}\equiv
\{(0,0),(1,0),(0,1),(1,1)\}$.
\end{example}

In Examples 2, we assume that the econometrician observes the aggregate
demand for all the respective bundles. We emphasize this point as it changes
the data requirement, and an interesting open question arises about what
happens if these requirements are not met. Examples of data sets that would
satisfy these requirements are when 1. individual observations are collected
through direct survey or scanner data on individual consumption (in every
market), 2. aggregate variables (market shares) are collected, but augmented
with a survey that asks individuals whether they consume each good
separately or as a bundle. 3. Finally, another possible data source are
producer's direct record of sales of bundles, provided each bundles are
recorded separately (e.g., when they are sold through promotional
activities). When discussing Example 2, we henceforth tacitly assume to have access to such data in
principle.

\subsection{Structural Demand}

\label{sec:inversion}

The first step toward identification of $f_\theta$ is to use a set of moment
conditions generated by instrumental variables to identify the aggregate
demand function $\varphi$. This section summarizes the identification result
obtained by BH (2014). Following BH (2014), we partition the covariates
as $X_{jt}=(X_{jt}^{(1)},X_{jt}^{(2)})\in \mathbb{R}\times\mathbb{R}%
^{d_X-1}, $ and make the following assumption.

\begin{ass}
\label{as:index} The coefficient $\beta^{(1)}_{ij}$ %on $X_{jt}^{(1)}$ 
is non-random for all $j$ and is normalized to 1.
\end{ass}

Assumption \ref{as:index} requires that at least one coefficient on the
covariates is non-random. Since we may freely choose the scale of utility,
we normalize the utility by setting $\beta^{(1)}_{ij}=1$ for all $j$. Under
Assumption \ref{as:index}, the utility for product $j$ can be written as $%
U^*_{jt}=X_{jt}^{(2)^{\prime }}{\beta_{ij}}^{(2)}+\alpha_{ij}
P_{jt}+\sigma_\epsilon \epsilon_{ijt}+D_{jt}$, where $D_{jt}\equiv
X_{jt}^{(1)}+\Xi_{jt}$ is the part of the utility that is common across
individuals. Assumption \ref{as:index} (i) is arguably strong but will
provide a way to obtain valid instruments required to identify the
structural demand (see BH, 2014, Section 7 for details). Under this
assumption, $U^*_{ijt}$ is strictly increasing in $D_{jt}$ but unaffected by 
$D_{kt}$ for all $k\ne j.$ In Example 1, together with a mild regularity
condition, this is sufficient for inverting the demand system to obtain $%
\Xi_t$ as a function of the market shares $S_t$, price $P_t$, and exogenous
covariates $X_t$ (Berry, Gandhi, and Haile, 2013). In what follows, we
redefine the aggregate demand as a function of $(X_t^{(2)},P_t,D_t)$ instead
of $(X_t,P_t,\Xi_t)$ by 
\begin{equation*}
    \phi(X_t^{(2)},P_t,D_t)\equiv
\varphi(X_t,P_t,\Xi_t),
\end{equation*} where $X_t=(X_t^{(1)},X_t^{(2)})$ and $%
D_t=\Xi_t+X^{(1)}_t$ and make the following assumption

\begin{ass}
\label{as:inv1} For some subset $\tilde{\mathbb{L}}$ of $\mathbb{L}$ whose
cardinality is $J$, there exists a unique function $\psi:\mathbb{R}^{J\times
(d_X-1)}\times\mathbb{R}^J\times\mathbb{R}^{J}\to\mathbb{R}^{J}$ such that $%
D_{jt}=\psi_j(X_t^{(2)},P_t,\tilde S_t)$ for $j=1,\cdots, J$, where $\tilde
S_t$ is a subvector of $S_t$, which stacks the components of $S_t$ whose
indices belong to $\tilde{\mathbb{L}}$.
\end{ass}

Under Assumption \ref{as:inv1}, we may write 
\begin{equation}
\Xi _{jt}=\psi _{j}(X_{t}^{(2)},P_{t},\tilde{S}_{t})-X_{jt}^{(1)}.
\label{eq:invxi}
\end{equation}%
This can be used to generate moment conditions in order to identify the
aggregate demand function. 

\setcounter{example}{0}

\begin{example}[BLP, continued]
\textrm{\textrm{Let $\tilde{\mathbb{L}}=\{1,\cdots ,J\}$. In this setting,
the inversion discussed above is the standard Berry inversion. A key
condition for the inversion is that the products are \emph{connected
substitutes} (Berry, Gandhi, and Haile (2013)). The linear random
coefficient specification as in \eqref{eq:randomu} is known to satisfy this
condition. Then, Assumption \ref{as:inv1} follows. } }
\end{example}

In Example \ref{ex:bundles}, one may employ an alternative inversion
strategy to obtain $\psi$ in \eqref{eq:invxi} using only subsystems of
demand such as $\tilde{\mathbb{L}}=\{(1,0),(1,1)\}$ or $\tilde{\mathbb{L}}%
=\{(0,0),(0,1)\}$. We defer details on this case to Section \ref{sec:bundles}%
.

The inverted system in \eqref{eq:invxi}, together with the following
assumption, yields a set of moment conditions the researcher can use to
identify the structural demand.

\begin{ass}
\label{as:npiv1} There is a vector of instrumental variables $Z_{t}\in%
\mathbb{R}^{d_Z}$ such that (i) $0=E[\Xi _{jt}|Z_{t},X_{t}]$, a.s.; (ii) for
any $B:\mathbb{R}^{Jk_{2}}\times \mathbb{R}^{J}\times \mathbb{R}%
^{J}\rightarrow \mathbb{R}$ with $E[|B(X_{t}^{(2)},P_{t},\tilde{S}%
_{t})|]<\infty $, it holds that%
\begin{equation*}
E[B(X_{t}^{(2)},P_{t},\tilde{S}_{t})|Z_{t},X_{t}]=0\Longrightarrow
B(X_{t}^{(2)},P_{t},\tilde{S}_{t})=0,~a.s.
\end{equation*}
\end{ass}

Assumption \ref{as:npiv1} (i) is a mean independence assumption on $\Xi
_{jt} $ given a set of instruments $Z_{t}$, which also normalizes the
location of $\Xi _{jt}$. Assumption \ref{as:npiv1} (ii) is a completeness
condition, which is common in the nonparametric IV literature, see BH\
(2014) for a detailed discussion. However, the role it plays here is
slightly different, as the moment condition leads to an integral equation
which is different from nonparametric IV (Newey and Powell 2003), and more
resembles GMM. 
In the online supplement Section \ref{sec:fullind}, we discuss an approach based on a
strengthening of the mean independence condition to full independence. In
case such a strengthening is economically palatable, we still retain the sum 
$X_{jt}^{(1)}+\Xi _{jt}$. Where $X_{jt}^{(1)}$ is similar to a
dependent variable in nonparametric IV.

Given Assumption \ref{as:npiv1} and \eqref{eq:invxi}, the unknown function $%
\psi $ can be identified through the following conditional moment
restrictions: 
\begin{equation}
E[\psi
_{j}(X_{t}^{(2)},P_{t},S_{t})-X_{jt}^{(1)}|Z_{t},X_{t}]=0,~~j=1,\cdots ,J.
\end{equation}%
We here state this result as a theorem. It is essentially Theorem 1 of BH (2014). 

\begin{theorem}
\label{thm:npiv1} Suppose Assumptions \ref{as:index}-\ref{as:npiv1} hold.
Then, $\psi$ is identified.
\end{theorem}

Once $\psi$ is identified, the structural demand $\phi$ can be identified
nonparametrically in Examples 1 and 2.

\setcounter{example}{0}

\begin{example}[Multinomial choice, continued]
\textrm{\textrm{\textrm{Recall that $\psi$ is a unique function s.t.
\begin{align}
S_{jt}=\phi_j(X^{(2)}_t,P_t,D_t),~j=1,\cdots, J~\Leftrightarrow
~\Xi_{jt}=\psi_j(X^{(2)}_t,P_t,\tilde S_t)-X^{(1)}_{jt},~j=1,\cdots, J,
\label{eq:blpeq}
\end{align}
where $\tilde S_t=(S_{1t},\cdots, S_{Jt}).$ Hence, the structural demand $%
(\phi_1,\cdots,\phi_J)$ is identified by Theorem \ref{thm:npiv1} and the
equivalence relation above. In addition, $\phi_0$ is identified through the
identity: $\phi_0=1-\sum_{j=1}^J\phi_j$. } } }
\end{example}

\begin{example}[Bundles, continued]
\textrm{\textrm{\textrm{Let $\tilde{\mathbb{L}}=\{(1,0),(1,1)\}$. Then $\psi $ is a unique function such that.
\begin{equation*}
S_{lt}=\phi _{l}(X_{t}^{(2)},P_{t},D_{t}),~l\in \tilde{\mathbb{L}}%
~~~\Leftrightarrow ~~~\Xi _{jt}=\psi _{j}(X_{t}^{(2)},P_{t},\tilde{S}%
_{t})-X_{t}^{(1)},~~j=1,2,
\end{equation*}%
where $\tilde{S}_{t}=(S_{(1,0),t},S_{(1,1),t}).$ Theorem \ref{thm:npiv1} and
the equivalence relation above then identify the demand for bundles $(1,0)$
and $(1,1)$. This, therefore, only identifies subcomponents of $\phi $.
Although these subcomponents are sufficient for recovering the random
coefficient density, one may also identify the rest of the subcomponents by
taking $\tilde{\mathbb{L}}=\{(0,0),(0,1)\}$ and applying Theorem \ref%
{thm:npiv1} again. } } }
\end{example}

\section{Identification of the Random Coefficient Density}

\label{sec:rcid}

This section contains the main innovation in this paper: We establish that
the density of random coefficients in the market-level demand models is
nonparametrically identified. Our strategy for identification of the random
coefficient density is to construct a function from the structural demand,
which is related to the density through an integral transform known as the 
\emph{Radon transform}. More precisely, we construct a function $\Phi (w,u)$
such that 
\begin{equation}
\frac{\partial \Phi (w,u)}{\partial u}=\mathcal{R}[f](w,u)~,
\label{eq:radonrep}
\end{equation}%
where $f$ is the density of interest, $w$ is a vector in $\mathbb{R}^{q}$
(with $q$ the dimension of the random coefficients), normalized to have unit
length, and $u\in \mathbb{R}$ is a scalar. In what follows, we let $\mathbb{S%
}^{q}\equiv \{v\in \mathbb{R}^{q}:\Vert v\Vert =1\}$ denote the unit sphere
in $\mathbb{R}^{q}$. $\mathcal{R}$ is the Radon transform defined pointwise
by 
\begin{equation}
\mathcal{R}[f](w,u)=\int_{P_{w,u}}f(v)d\mu _{w,u}(v).  \label{eq:radon}
\end{equation}%
where $P_{w,u}$ denotes the hyperplane $\{v\in \mathbb{R}^{q}:v^{\prime
}w=u\},$ and $\mu _{{w,u}}$ is the Lebesgue measure on $P_{w,u}$. See for
example Helgason (1999) for details on the properties of the Radon transform
including its injectivity. Our identification strategy is constructive and
will therefore suggest a natural nonparametric estimator. Applications of
the Radon transform to random coefficients models have been studied in
Beran, Feuerverger, and Hall (1996), Hoderlein, Klemel\"{a}, and Mammen
(2010), and Gautier and Hoderlein (2015).

Throughout, we maintain the following assumption.

\begin{ass}
\label{as:covariates} (i) For all $j\in\{1,\cdots, J\}$, $( X^{(2)}_{jt},
P_{jt}, D_{jt})$ are absolutely continuous with respect to Lebesgue measure
on $\mathbb{R}^{d_X-1}\times \mathbb{R}\times\mathbb{R}$; (ii) the random
coefficients $\theta$ are independent of $(X_{t},P_t,D_t)$. 
\end{ass}

Assumption \ref{as:covariates} (i) requires that $%
(X_{jt}^{(2)},P_{jt},D_{jt})$ are continuously distributed for all $j$. By
Assumption \ref{as:covariates} (ii), we assume that the covariates $%
(X_{t},P_{t},D_{t})$ are exogenous to the individual heterogeneity. These
conditions are used to invert the Radon transform.

Before proceeding further, we overview our identification strategy in
relation to the key differences between the BLP and pure characteristics
models. Heuristically, for a given $(w,u)\in\mathbb{S}^q\times\mathbb{R}$,
the Radon transform aggregates individuals whose coefficients are on the
hyperplane $P_{w,u}$. For each $(w,u)$, we relate this aggregate value to a
feature of the demand with a specific product characteristics. By varying $%
(w,u)$ and inverting the map $\mathcal{R}$ in \eqref{eq:radonrep}, we may
then recover the distribution of the random coefficients. A key step in this
identification argument is the construction of a function $\Phi$ satisfying %
\eqref{eq:radonrep}. The two demand models suggest different strategies to
construct $\Phi.$ In the BLP model, we construct $\Phi$ for each product $j$
and recover the joint distribution of the coefficients $(\beta^{(2)}_{it},%
\alpha_{it},\epsilon_{ijt})$. We take this approach because the presence of
the tastes for products requires us to isolate the demand for each product
from the rest. On the other hand, the pure characteristics model does not
require such an approach. Furthermore, both models allow the researcher to 
combine demand across different markets to construct $\Phi$.

\subsection{BLP model}

Throughout this section, we let $\sigma_\epsilon=1$.\footnote{%
Here, the scale of the taste for product $\epsilon_{ijt}$ is normalized
relative to the scale of $X^{(1)}_{jt}$ as we set the coefficient on $%
X^{(1)}_{jt}$ to 1 in Assumption \ref{as:index}.} Recall that the demand for
good $j$ with the product characteristics $(X_{t},P_{t},\Xi _{t})$ is as
given in \eqref{eq:blp1}. Since $D_{t}=X_{t}^{(1)}+\Xi _{t}$, the demand in
market $t$ with $(X_{t}^{(2)},P_{t},D_{t})=(x^{(2)},p,\delta )$ is given by: 
\begin{align}\label{eq:multnom}
\begin{split}
\phi _{j}(x^{(2)},p,\delta )=&\int 1\{x_{j}^{(2)}{}^{\prime }{b}%
^{(2)}+ap_{j}+e_j>-\delta _{j}\}\\
&\times 1\{(x_{j}^{(2)}-x_{1}^{(2)}){}^{\prime }{b}%
^{(2)}+a(p_{j}-p_{1})+(e_j-e_1)>-(\delta _{j}-\delta _{1})\} \\
\cdots 1\{(x_{j}^{(2)}-&x_{J}^{(2)})^{\prime }{b}%
^{(2)}+a(p_{j}-p_{J})+(e_j-e_J)>-(\delta _{j}-\delta _{J})\}f_{\theta
}(b^{(2)},a,e)d\theta ~.  
\end{split}
\end{align}
Suppose the vertical characteristics $\{D_{kt},k\ne j\}$ (for products other
than $j$) have a large enough support so that  $%
(X_{jt}^{(2)}-X_{kt}^{(2)})^{\prime }{\beta}^{
(2)}_{it}+\alpha_{it}(P_{jt}-P_{kt})+(\epsilon_{ijt}-%
\epsilon_{ikt})-D_{jt}>D_{kt}$ for all $k\ne j$ for some values of $%
D_{kt},k\ne j$. The demand for good $j$ for such values of $D_{kt},k\ne j$
is then 
\begin{align}  \label{eq:blp_phitilde}
\begin{split}
\tilde\Phi_j(x_{j}^{(2)},p_{j},\delta_{j})&=\lim_{\delta_1,\ldots,\delta_{j-1},%
\delta_{j+1},\ldots, \delta_J \rightarrow -\infty} \phi_j(x^{(2)},p,\delta)\\
&=\int 1\{x_{j}^{(2)}{}^{\prime }{b}^{
(2)}+ap_{j}+e_j<-\delta_{j}\}f_{\vartheta_j }(b^{(2)},a,e_j)d\vartheta_j,
\end{split}
\end{align}
where $f_{\vartheta_j }$ is the joint density of the subvector $%
\vartheta_{ijt}\equiv(\beta^{(2)}_{it},\alpha_{it},\epsilon_{ijt} )$ of the
random coefficients. Let $w\equiv (x_{j}^{(2)},p_{j},1)/\Vert
(x_{j}^{(2)},p_{j},1)\Vert $ and $u\equiv \delta_j/\Vert
(x_{j}^{(2)},p_{j},1)\Vert $. Define 
\begin{align}\label{eq:Phi}
\begin{split}
\Phi(w,u)&\equiv\tilde\Phi_j\Bigg(\frac{x_j^{(2)}}{\Vert
(x_{j}^{(2)},p_{j},1)\Vert},\frac{p_j}{\Vert (x_{j}^{(2)},p_{j},1)\Vert},%
\frac{\delta_j}{\Vert (x_{j}^{(2)},p_{j},1)\Vert}\Bigg) \\
&=\tilde\Phi_j(x_{j}^{(2)},p_{j},\delta _{j}),~~(x_{j}^{(2)},p_{j},\delta
_{j})\in \mathrm{supp}\,(X_{jt}^{(2)},P_{jt},D_{jt}),  
\end{split}
\end{align}
where the second equality holds because normalizing the scale of $%
(x_{j}^{(2)},p_{j},\delta _{j})$ does not change the value of $\tilde\Phi_j$%
. $\Phi$ then satisfies 
\begin{align}\label{eq:sum1}
\begin{split}
\Phi(w,&u )=-\int 1\{w^{\prime }\vartheta_j<-u \}f_{\vartheta_j
}(b^{(2)},a,e_j)d\vartheta_j \\
&=-\int_{-\infty }^{-u}\int_{P_{w,r}}f_{\vartheta_j }(b^{(2)},a,e_j)d\mu _{{%
w,r}}(b^{(2)},a,e_j)dr=-\int_{-\infty }^{-u}\mathcal{R}[f_{\vartheta_j
}](w,r)dr~,
\end{split}
\end{align}
Hence, by taking a derivative with respect to $u$, we may relate $\Phi$ to
the random coefficient density through the Radon transform: 
\begin{equation}
\frac{\partial \Phi(w,u)}{\partial u}=\mathcal{R}[f_{\vartheta_j }](w,u).
\label{eq:radon1}
\end{equation}%
Note that since the structural demand $\phi $ is identified by Theorem \ref%
{thm:npiv1}, $\Phi$ is nonparametrically identified as well. Hence, Eq. %
\eqref{eq:radon1} gives an operator that maps the random coefficient density
to an object identified by the moment condition studied in the previous
section. To construct $\Phi$ described above and to invert the Radon
transform, we formally make the following assumptions. Below, for each $%
1\le,j,k\le J$, we let $V_{jk} =(X_{jt}^{(2)}-X_{kt}^{(2)})^{\prime }{\beta}%
^{
(2)}_{it}+\alpha_{it}(P_{jt}-P_{kt})+(\epsilon_{ijt}-\epsilon_{ikt})-D_{jt}$
and make the following assumptions on the support of the product
characteristics.\footnote{$V_{jk}$ is a random variable that varies across
individuals and markets and hence should be denoted as $V_{ijkt}$ in
principle. For conciseness, we drop subscripts $i$ and $t$ below.}

\begin{ass}
\label{as:rc_blp} Let $\mathcal{J}$ be a nonempty subset of $\{1,\cdots,J\}.$
For each $j\in\mathcal{J}$, let $\mathrm{supp}\,(V_{jk},k\ne j)\subset \mathrm{%
supp}\,(D_{kt},k\ne j)$.
\end{ass}

Assumption \ref{as:rc_blp} requires that one may vary the vertical
characteristics of the alternative products $\{D_{kt},k\neq j\}$ on a large
enough support so that the demand for product $j$ is determined through its
choice between product $j$ and the outside good. This identification
argument therefore uses a \textquotedblleft thin\textquotedblright\
(lower-dimensional) subset of the support of the covariates, which is due to
the presence of the tastes for products. This is in remarkable contrast with
the identification of the random coefficients density in the PCM (analyzed in the next section) which does
not rely on thin sets. 

\begin{ass}
\label{as:analytic} One of the following conditions hold

\begin{enumerate}
\item[(i)] $\bigcup_{j\in \mathcal{J}}\mathrm{supp}\,(X_{jt}^{(2)},P_{jt},D_{jt})$
has full support in $\mathbb{R}^{d_X-1}\times\mathbb{R}\times\mathbb{R}$.

\item[(ii)] $\bigcup_{j\in \mathcal{J}}\mathrm{supp}\,(X_{jt}^{(2)},P_{jt})$
contains an open ball $B_\mathcal{J} \subset \mathbb{R}^{d_X-1}\times\mathbb{%
R}$. For every $(x,p) \in B_\mathcal{J}$ and every $(b^{(2)},a,e_j) \in 
\mathrm{supp}\,(\vartheta_j)$ it holds that 
\begin{align}
\big(x,p,-x^{\prime} b^{(2)} - ap - e_j\big) \in \bigcup_{j\in \mathcal{J}}%
\mathrm{supp}\,\left(X_{jt}^{(2)},P_{jt},D_{jt}\right).  \label{eq:supp_blp}
\end{align}
Furthermore, all the absolute moments of each component of $\theta_{it}$ are
finite, and for any fixed $z\in \mathbb{R}_+$, $\lim_{l\to\infty}\frac{z^l%
}{l!} E[(|\theta_{it}^{(1)}|+\cdots+|\theta^{(d_\theta)}_{it}|)^l]=0$. 
\end{enumerate}
\end{ass}

Assumption \ref{as:analytic} (i) is our benchmark assumption. Under this
assumption, no restrictions on $\theta_{it}$ are necessary for
identification. In fact, the identification strategy would be valid for
arbitrary Borel measures and may also be applied to settings where $%
\theta_{it}$ does not have a density.\footnote{%
More precisely, the Radon transform $\mathcal{R}[f_{\vartheta _{j}}](w,u)$
gives $f_{\vartheta _{j}}$'s integral along each hyperplane $P_{w,u}=\{v\in 
\mathbb{R}^{d_{\theta }}:v^{\prime }w=u\}$ defined by the \emph{angle} $%
w=(x_{j}^{(2)},p_{j},1)/\Vert (x_{j}^{(2)},p_{j},1)\Vert $ and \emph{offset} 
$u=\delta _{j}/\Vert (x_{j}^{(2)},p_{j},1)\Vert $. For recovering $%
f_{\vartheta _{j}}$ from its Radon transform, one needs exogenous variations
in both. Our proof uses the fact that varying $w$ over the hemisphere $%
\mathbb{H}_{+}\equiv \{w=(w_{1},w_{2},\cdots ,w_{d_{\vartheta _{j}}})\in 
\mathbb{S}^{d_{\vartheta _{j}}-1}:w_{d_{\vartheta _{j}}}\geq 0\}$ and $u$
over $\mathbb{R}$ suffices to recover $f_{\vartheta _{j}}$.} However, this
large support assumption is stringent and may be violated by various product
characteristics and prices used in practice. Hence, it should be viewed as a
benchmark to understand what the model requires to identify the distribution 
$f_\theta$ of $\theta_{it}$ if one does not impose any restriction on it.

Assumption \ref{as:analytic} (ii) is an alternative condition, which relaxes
the support requirement significantly. Instead of a large support, it is
enough for the product characteristics to have a properly combined support
that contains a (possibly small) open ball $B_{\mathcal{J}}$ in it. This
includes as a special case where a single product's characteristics $%
(X_{jt}^{(2)},P_{jt})$ %,D_{jt})$ 
contains an open ball, which can be met in various applications. Even if
such a product does not exist, identification of the random coefficient
density is possible as long as the required support condition is met by
combining the supports of multiple products belonging to $\mathcal{J}$. This
means that our identification strategy may use variations of $%
(X_{jt}^{(2)},P_{jt})$ across products. To illustrate, consider three
products $J=3$. If $(D_{2t},D_{3t})$ have a large support in the sense of
Assumption \ref{as:rc_blp} ($\mathcal{J}=\{1\}$ in this case),
identification of the random coefficient density is possible as long as the
characteristics of good 1 contains an open ball. If all $\{D_{jt}\}_{j=1}^3$
jointly have a large support (this implies $\mathcal{J}=\{1,2,3\}$), our
requirement on $(X_{jt}^{(2)},P_{jt})$ becomes even milder as we only need
to construct an open ball by combining the characteristics of all three products.

The condition in \eqref{eq:supp_blp} allows for a bounded support of $%
(X_{jt}^{(2)},P_{jt})$. Further, if $\vartheta _{j}$ has a bounded
support, Assumption \ref{as:analytic} (ii) will allow for a bounded support of $D_{j}$. The price to pay for this relaxation of the
support requirement is a regularity assumption on the moments of $\theta
_{it}$. This rules out heavy tailed distributions that are not determined by
their moments. A sufficient, yet stronger than necessary, condition for this
assumption is a compact support of $f_{\theta }$. Under Assumption \ref%
{as:analytic} (ii), the characteristic function $w\mapsto \varphi
_{\vartheta _{j}}(tw)$ of $\vartheta _{ijt}$ (a key element of the Radon
inversion) is analytic and thereby uniquely determined by its restriction to
a non-empty full dimensional subset of its domain.\footnote{%
This type of moment condition on $\theta _{it}$ is common in the recent
literature. See, for example, Hoderlein, Holzmann, and Meister (2014) and
Masten (2014).} 
Hence, $f_{\vartheta _{j}}$ can be identified if one varies $%
(X_{jt}^{(2)},P_{jt})$ on a full dimensional subset.

Under the conditions given in the theorem below, the Radon inversion
identifies $f_{\vartheta_j}$. If one is interested in the joint density of
the coefficients on the product characteristics $(\beta^{(2)}_{it},%
\alpha_{it})$, one may stop here as marginalizing $f_{\vartheta_j}$ gives
the desired density. The joint distribution of the coefficients including
the tastes for products can be identified under an additional independence
assumption. We state this result in the following theorem.

\begin{theorem}
\label{thm:fbeta} Suppose Assumptions \ref{as:index}-\ref{as:analytic} hold.
Suppose the conditional distribution of $\epsilon_{ijt}$ given $%
(\beta^{(2)}_{it},\alpha_{it})$ is identical for all $j\in\mathcal{J}$.
Then, (i) for each $j\in\mathcal{J}$, the density $f_{\vartheta_j}$ is
identified, where $\vartheta_{ijt}=(\beta^{(2)}_{it},\alpha_{it},%
\epsilon_{ijt})$; (ii) If, in addition, $\{\epsilon_{ijt},j\in\mathcal{J}\}$
are independently distributed (across $j$) conditional on $%
(\beta^{(2)}_{it},\alpha_{it})$, the joint density $f_{\theta_{\mathcal{J}}}$
of $\theta_{\mathcal{J}}=(\beta^{(2)}_{it},\alpha_{it},\{\epsilon_{ijt}\}_{j%
\in\mathcal{J}})$ is identified.
\end{theorem}

An immediate corollary is the following.

\begin{corollary}
\label{cor:fbeta}  Suppose Assumptions \ref{as:index}-\ref{as:analytic}
hold. Let $\{\epsilon_{ijt}\}_{j=1}^J$ be i.i.d. (across $j$)
conditional on $(\beta^{(2)}_{it},\alpha_{it})$. Then, the joint density $%
f_\theta$ of all random coefficients $\theta_{it}=(\beta^{(2)}_{it},%
\alpha_{it},\{\epsilon_{ijt}\}_{j=1}^J)$ is identified.
\end{corollary}

%Several remarks are in order.

\begin{remark}\rm
\textrm{Theorems \ref{thm:npiv1} and \ref{thm:fbeta} shed light on the roles
played by the key features of the BLP-type demand model: the invertibility
of the demand system, instrumental variables, and the linear random
coefficients specification. In Theorem \ref{thm:npiv1}, the invertibility
and instrumental variables play key roles in identifying the demand. Once
the demand is identified, one may ``observe'' the vector $%
(X^{(2)}_t,P_t,D_t) $ of product characteristics. This is possible because
the invertibility of demand allows one to recover the unobserved product
characteristics $\Xi_t$ from the market shares $S_t$ (together with other
covariates). One may then vary $(X^{(2)}_t,P_t,D_t)$ across markets in a
manner that is exogenous to the individual heterogeneity $\theta_{it}$.
Theorem \ref{thm:fbeta} and Corollary \ref{cor:fbeta} show that this
exogenous variation combined with the linear random coefficients
specification allows to trace out the distribution of $\theta_{it}$. }
\end{remark}

\begin{remark}\rm
\textrm{\label{rem:iid} The identical distribution assumption on the tastes
for products  in Theorem \ref{thm:fbeta} is compatible with commonly used
utility specifications and can also be relaxed at the cost of a stronger
support condition on the product characteristics. In applications, it is
often assumed that the utility of product $j$ is 
\begin{align}
U_{ijt}^{\ast }=\beta^0_{it}+\tilde X_{jt}^{\prime }\beta _{it}+\alpha
_{it}P_{jt}+\Xi _{jt}+\tilde \epsilon_{ijt},
\end{align}
where $\tilde X_{jt}$ is a vector of non-constant product characteristics, $%
\beta^0_{it}$ is an individual specific intercept, which measures the
utility difference between inside goods and the outside good, and $%
\tilde\epsilon_{ijt}$ is a mean zero error that follows the Type-I extreme
value distribution. The requirement that $\epsilon_{ijt}=\beta^0_{it}+\tilde%
\epsilon_{ijt}$ are i.i.d. across $j$ (conditional on $(\beta_{it},%
\alpha_{it})$) can be met if $\tilde\epsilon_{ijt}$ are i.i.d. across $j$. }

\textrm{If for each $j$, $(X_{jt}^{(2)},P_{jt},D_{jt})$ fulfills the support
condition in Assumption \ref{as:analytic} (i) or Assumption \ref{as:analytic}
(ii), one can drop the identical distribution assumption. This is because
one can identify $f_{\vartheta _{j}}$ for all $j$ by inverting the Radon
transform in \eqref{eq:radon1} repeatedly. This in turn implies that the
distribution of $\epsilon _{ijt}$ conditional on $(\beta _{it}^{(2)},\alpha
_{it})$ is identified for each $j$. If the tastes for products $%
\{\epsilon_{ijt}\}_{j=1}^J$ are mutually independent (conditional on $(\beta
_{it}^{(2)},\alpha _{it})$), as is commonly assumed in BLP, the joint density $f_{\theta }$ is identified. }

\textrm{Finally, we comment on what an additional parametric assumption may
add to our result. If one assumes that the tastes for products are i.i.d.
and follows a parametric distribution, Eq \eqref{eq:multnom} reduces to $%
\phi _{j}(x^{(2)},p,\delta )=\int L(x^{(2)}{}^{\prime
(2)}+ap+\delta)f_{(\beta,\alpha)}(b,a)dbda,$ for some function $L$, e.g. $L$
is the logit function when $\{\epsilon_{ijt}\}$ follows a Type-I extreme
value distribution. This type of integral equation is considered in Fox,
Kim, Ryan, and Bajari (2012) in the context of individual-level demand model
without endogeneity. Given that $\phi$ is identified, we believe that it is
possible to extend their framework to the market-level demand model with
endogeneity and identify $f_{(\beta,\alpha)}$ semiparametrically.  This
approach may allow us to relax some of the support conditions. To keep a
tight focus on nonparametric identification, we leave this extension
for future work.}
\end{remark}

\begin{remark}
\label{rem:support} \rm Our identification result reveals the nature of
the BLP-type demand model. A positive aspect of our result is that the
preference is nonparametrically identified if one observes full dimensional
variations in the consumers' choice sets (represented by $%
(X^{(2)}_{jt},P_{jt},D_{jt})$) across markets. The identifying power is
quite strong, if the product characteristics jointly span a full support,
i.e. $\bigcup_{j\in\mathcal{J}}(X^{(2)}_{jt},P_{jt},D_{jt})=\mathbb{R}%
^{d_X-1}\times\mathbb{R}\times\mathbb{R}.$ On the other hand, if the product
characteristics have limited variations, the identifying power of the model
on the distribution of preferences may be limited. In particular,
identification is not achieved only with discrete covariates. Hence, for
such settings, one needs to augment the model structure with a parametric
specification. Another interesting direction would be to conduct partial
identification analysis on functionals of $f_\theta$, while imposing weak
support restrictions. We leave this possibility for future research. 
\end{remark}

\subsection{Pure Characteristics Demand Models}\label{sec:PCM} 

Throughout this section, we consider the following utility
specification where each product's utility is fully determined by the tastes
for the product characteristics: 
\begin{equation}
U_{ijt}^{\ast }\equiv X_{jt}^{\prime }\beta _{it}+\alpha _{it}P_{jt}+\Xi
_{jt},~j=1,\cdots ,J~.  \label{eq:randomu_pure}
\end{equation}%
In other words, we set $\sigma _{\epsilon }=0$ in \eqref{eq:randomu}. For
this model, we employ a different, and arguably less restrictive, strategy
from the one adopted in the previous section to construct $\Phi $ in %
\eqref{eq:radonrep}. Below, we maintain Assumptions \ref{as:index}-\ref%
{as:npiv1}, which ensure the identification of demand by Theorem \ref%
{thm:npiv1}. The demand for good $j$ with the product characteristics $(X_{t},P_{t},\Xi
_{t})$ is as given in \eqref{eq:blp1} but with $\sigma_\epsilon=0$. Since $%
D_{t}=X_{t}^{(1)}+\Xi _{t}$, the demand in market $t$ with $%
(X_{t}^{(2)},P_{t},D_{t})=(x^{(2)},p,\delta )$ is given by: 
\begin{multline}
\phi _{j}(x^{(2)},p,\delta )=\int 1\{x_{j}^{(2)}{}^{\prime }{b}%
^{(2)}+ap_{j}>-\delta _{j}\}1\{(x_{j}^{(2)}-x_{1}^{(2)}){}^{\prime }{b}%
^{(2)}+a(p_{j}-p_{1})>-(\delta _{j}-\delta _{1})\} \\
\cdots 1\{(x_{j}^{(2)}-x_{J}^{(2)}){}^{\prime }{b}^{(2)}+a(p_{j}-p_{J})>-(%
\delta _{j}-\delta _{J})\}f_{\theta }(b^{(2)},a)d\theta ~.
\end{multline}

For any subset $\mathcal{J}$ of $\{1,\cdots ,J\}\setminus \{j\}$, let $%
\mathcal{M}_{\mathcal{J}}$ denote the map $(x^{(2)},p,\delta )\mapsto (%
\acute{x}^{(2)},\acute{p},\acute{\delta})$ that is uniquely defined by the
following properties: 
\begin{align*}
(\acute{x}_{j}^{(2)}-\acute{x}_{i}^{(2)},\acute{p}_{j}-\acute{p}_{i},\acute{%
\delta}_{j}-\acute{\delta}_{i})&
=-(x_{j}^{(2)}-x_{i}^{(2)},p_{j}-p_{i},\delta _{j}-\delta _{i}),~\forall
i\in \mathcal{J}~, \\
(\acute{x}_{i}^{(2)},\acute{p}_{i},\acute{\delta}_{i})&
=(x_{i}^{(2)},p_{i},\delta _{i}),~\forall i\notin \mathcal{J}~.
\end{align*}%
In words, for a given product $j$ and product characteristics $(x^{(2)},p,\delta )$, this map finds another value $(%
\acute{x}^{(2)},\acute{p},\acute{\delta})$ of  product characteristics such that, for products $i$ belonging to $\mathcal J$, the difference in the product characteristics (e.g. $\acute{x}_{j}^{(2)}-\acute{x}_{i}^{(2)}$) coincides with the original value (e.g., $x_{j}^{(2)}-x_{i}^{(2)}$) in terms of magnitude but has an opposite sign. For products $i$ not belonging to $\mathcal J$, the map sets their product characteristics the original value $(x^{(2)}_i,p_i,\delta_i)$.

Consider the
composition $\phi _{j}\circ \mathcal{M}_{\mathcal{J}}(x^{(2)},p,\delta )$.
If $( \acute{x}^{(2)},\acute{p},\acute{\delta})$ is in the support, this
corresponds to the demand of product $j$ in some market (say $t^{\prime }$)
with $(X_{t^{\prime }}^{(2)},P_{t^{\prime }},D_{t^{\prime }})=( \acute{x}%
^{(2)},\acute{p},\acute{\delta})$. We then define 
\begin{equation}
\tilde{\Phi}_{j}(x_{j}^{(2)},p_{j},\delta _{j})\equiv -\sum_{\mathcal{J}%
\subseteq \{1,\cdots J\}\setminus \{j\}}\phi _{j}\circ \mathcal{M}_{\mathcal{%
J}}(x^{(2)},p,\delta )~.  \label{eq:sum_pure}
\end{equation}%
Eq \eqref{eq:sum_pure} aggregates the structural demand function for good $j$
in different markets to define a function, which can be related to the
random coefficient density in a simple way. This operation can be easily
understood when $J=2$, where for example demand for product 1 is given by 
\begin{multline*}
\phi _{1}(x^{(2)},p,\delta )=\int 1\{x_{1}^{(2)}{}^{\prime }{b}%
^{(2)}+ap_{1}<-\delta _{1}\} \\
\times 1\{(x_{1}^{(2)}{}-x_{2}^{(2)}){}{}^{\prime }{b}%
^{(2)}+a(p_{1}-p_{2})<-(\delta _{1}-\delta _{2})\}f_{\theta
}(b^{(2)},a)d\theta ~.
\end{multline*}%
Then, $\tilde{\Phi}_{1}$ is given by 
\begin{align}\label{eq:sum1_pure}
\begin{split}
&\tilde{\Phi}_{1}(x_{1}^{(2)},p_{1},\delta _{1}) =-\phi _{1}\circ \mathcal{M}%
_{\emptyset }(x^{(2)},p,\delta)-\phi _{1}\circ \mathcal{M}%
_{\{2\}}(x^{(2)},p,\delta) \\
& =-\int 1\{x_{1}^{(2)}{}^{\prime }b^{(2)}+ap_{1}<-\delta _{1}\}\Big(%
1\{(x_{1}^{(2)}-x_{2}^{(2)})^{\prime }b^{(2)}+a(p_{1}-p_{2})<-(\delta
_{1}-\delta _{2})\}\\
& \hspace{76pt} +1\{(x_{1}^{(2)}{}-x_{2}^{(2)})^{\prime
}b^{(2)}+a(p_{1}-p_{2})>-(\delta _{1}-\delta _{2})\}\Big)f_{\theta
}(b^{(2)},a)d\theta \\
& =-\int 1\{x_{1}^{(2)}{}^{\prime }b^{(2)}+ap_{1}<-\delta _{1}\}f_{\theta
}(b^{(2)},a)d\theta  
\end{split}
\end{align}%
This shows that aggregating the demand in the two markets with $%
(X_{t}^{(2)},P_{t},D_{t})=(x^{(2)},p,\delta )$ and $(X_{t^{\prime
}}^{(2)},P_{t^{\prime }},D_{t^{\prime }})=(\acute{x}^{(2)},\acute{p},\acute{%
\delta})$ yields a function $\tilde{\Phi}_{1}$ that depends only on product
1's characteristic $(x_{1}^{(2)},p_{1},\delta _{1})$ through a single index
in \eqref{eq:sum1_pure}. This then allows us to trace out the random
coefficients density by varying product 1's characteristic as done in the
BLP model. Since the operation above yields a function that depends only on
the characteristic of a single product, we call it \emph{marginalization of
demand}.\footnote{%
Note however that this marginalized demand still depends on the joint
distribution of the entire random coefficient vector.}

Eq. \eqref{eq:sum_pure} generalizes this argument to settings with $J\geq 2$%
. For the marginalization of demand to work, the product characteristic $(%
\acute{x}^{(2)},\acute{p},\acute{\delta})=\mathcal{M}_{\mathcal{J}}(x^{(2)},
p,\delta)$ needs to be an observable value, meaning it must be in the
support. Formally, a value of the product characteristic $(x^{(2)},
p,\delta)\in \text{supp}( X^{(2)}_t, P_{t}, D_{t})$ is said to permit \emph{%
marginalization of demand with respect to product $j$} if 
\begin{align}
\mathcal{M}_{\mathcal{J}}( x^{(2)}, p,\delta)\in \text{supp}( X^{(2)}_t,
P_{t}, D_{t}),~\forall \mathcal{J}\subseteq\{1,\cdots, J\}\setminus\{j\}.
\end{align}

As done in the BLP setting, we will only require that a rich enough set to
recover $f_\theta$ can be constructed by combining the supports of multiple
products' characteristics. Toward this end, for each $j\in \{1,\cdots,J\}$,
let $\pi_j$ be the projection map such that $(x^{(2)}_{j},p_j,\delta_j)=%
\pi_j(x^{(2)}, p,\delta)$, and define the following sets: 
\begin{align*}
\mathcal{H}_j&\equiv \{(x^{(2)}, p,\delta)\in \text{supp}( X^{(2)}_t,
P_{t},D_{t}): \mathcal{M}_{\mathcal{J}}( x^{(2)}, p,\delta)\in \text{supp}%
(X^{(2)}_t, P_{t}, D_{t}), \\
&\hspace{0.3in}\text{ for all }\mathcal{J}\subseteq\{1,\cdots,J\}\setminus%
\{j\}\}, \\
\mathcal{S}_j&\equiv\{(x^{(2)}_{j},p_j,\delta_j)\in \text{supp}(
X^{(2)}_{jt}, P_{jt},D_{jt}):(x^{(2)}_{j},p_j,\delta_j)=\pi_j(x^{(2)},
p,\delta),\\
&\hspace{0.3in}~\text{ for some }(x^{(2)}, p,\delta)\in \mathcal{H}_j\}.
\end{align*}
In words, $\mathcal{H}_j$ is the set of the entire product characteristic
vectors for which marginalization with respect to product $j$ is permitted. $%
\mathcal{S}_j$ is the coordinate projection of $\mathcal{H}_j$ onto the
space of product $j$'s characteristics. We then make the following
assumption.

\begin{ass}
\label{as:support_pure} One of the following conditions hold:
\begin{enumerate}
\item[(i)] $\bigcup_{j=1}^J\mathcal{S}_j=\mathbb{R}^{d_X-1}\times\mathbb{R}\times 
\mathbb{R}$;

\item[(ii)] $\bigcup_{j=1}^J\mathcal{S}_j=\mathbb{E}\times \mathbb{D}$, where $%
\mathbb{E}$ contains an open ball $B \subset \mathbb{R}^{d_X-1}\times\mathbb{%
R}$, and $\mathbb{D}\subseteq \mathbb{R}$. For every $(x,p) \in B$ and every 
$(b^{(2)},a) \in \mathrm{supp}\,(\theta_{it})$, it holds that $%
(x,p,-x^{\prime} b^{(2)} - ap) \in \bigcup_{j=1}^J\mathcal{S}_j$.

Furthermore, all the absolute moments of each component of $\theta_{it}$ are
finite, and for any fixed $z\in \mathbb{R}_+$, $0=\lim_{l\to\infty}\frac{z^l%
}{l!}(E[|\theta_{it}^{(1)}|^l]+\cdots+E[|\theta^{(d_\theta)}_{it}|^l])$.
\end{enumerate}
\end{ass}

The idea behind Assumption \ref{as:support_pure} is as follows. For the
moment, suppose we don't impose any moment condition on the random
coefficient density. Also, fix a benchmark product $j$. For any $%
(x^{(2)}_{j},p_j,\delta_j)\in\mathcal{S}_j$, one may find a vector $%
(x^{(2)}, p,\delta)$ of all product characteristics for which
marginalization of demand is allowed. Then, one would wish to vary $%
(x^{(2)}_{j},p_j,\delta_j)$ to trace out the random coefficient density. 
This is possible, of course, if $\mathcal{S}_j=\mathbb{R}^{d_X-1}\times%
\mathbb{R}\times\mathbb{R}$, meaning that marginalization is possible
everywhere with respect to product $j.$ However, this assumption may be too
strong in empirical applications. One may not be able to find any single
product, for which this condition is satisfied.  Assumption \ref%
{as:support_pure} (i) relaxes this requirement substantially using the
structure of the model. Observe that the identification argument is
symmetric across products because only the characteristics matter. Hence,
the argument is valid as long as, for each $(\mathbf{x}^{(2)},\mathbf{p},%
\mathbf{d})\in \mathbb{R}^{d_X-1}\times\mathbb{R}\times \mathbb{R}$, one can
find \emph{some} product for which marginalization is permitted. This is the
reason why it is enough to ``patch'' $\mathcal{S}_j$s together to $%
\mathbb{R}^{d_X-1}\times\mathbb{R}\times \mathbb{R}$ in Assumption \ref%
{as:support_pure} (i). This condition can be made even weaker with the help
of an additional moment condition. In Assumption \ref{as:support_pure} (ii),
we only require that $\mathcal{S}_j$s combined together contain an open ball (in terms of $(\mathbf{x}^{(2)},\mathbf{p})$).
This support requirement is quite mild, and hence it can be satisfied even
if each product's characteristic has limited variation across markets. Note
also that, if $\mathrm{supp}\,(\theta_{it})$ is compact, the support of $%
D_{jt}$ can be compact as well.

It is important to note that we construct $\tilde\Phi_j$ without relying on
any ``thin'' (lower-dimensional) subset of the support of the product
characteristics as done in the BLP model. Instead, we construct $\tilde\Phi_j
$ in \eqref{eq:sum_pure} by combining the demand in different markets. 
This is desirable as estimators that rely on thin or irregular
identification may have a slow rate of convergence (Khan and Tamer, 2010).
In the pure characteristics model, the individuals have varying tastes
(random coefficients) over the product characteristics but not over the
products themselves. This is the key feature of the model that allows us to
identify the random coefficients through the variation of the product
characteristics $( X^{(2)}_t, P_{t}, D_{t})$. In contrast, in the BLP model,
there was an additional taste for the product itself, which was the main
reason for using the thin set to isolate the demand for each product.

Given Assumption \ref{as:support_pure}, we now construct $\Phi$ in Eq. %
\eqref{eq:radonrep}. For each $(\mathbf{x}^{(2)},\mathbf{p},\mathbf{d})\in
\bigcup_{j=1}^J\mathcal{S}_j$, let $w\equiv (\mathbf{x}^{(2)},\mathbf{p}%
)/\Vert (\mathbf{x}^{(2)},\mathbf{p})\Vert $ and $u\equiv \mathbf{d}/\Vert (%
\mathbf{x}^{(2)},\mathbf{p})\Vert $. Define 
\begin{align}
\Phi(w,u)\equiv\tilde\Phi_j\Big(\frac{\mathbf{x}^{(2)}}{\Vert (\mathbf{x}%
^{(2)},\mathbf{p})\Vert},\frac{\mathbf{p}}{\Vert (\mathbf{x}^{(2)},\mathbf{p}%
)\Vert},\frac{\mathbf{d}}{\Vert (\mathbf{x}^{(2)},\mathbf{p})\Vert}\Big),~%
\text{where}~(\mathbf{x}^{(2)},\mathbf{p},\mathbf{d})\in \mathcal{S}_j.
\label{eq:Phi_pure}
\end{align}
Here, for each $(\mathbf{x}^{(2)},\mathbf{p},\mathbf{d})$, any $j$ can be
used to construct $\tilde\Phi_j$ through marginalization as long as $%
\mathcal{S}_j$ contains $(\mathbf{x}^{(2)},\mathbf{p},\mathbf{d})$. Then $%
\Phi$ is defined on a set that is rich enough to invert the Radon (or
limited angle Radon) transform. The rest of the analysis parallels our
analysis of the BLP model.\footnote{%
Note that the additional independence (or i.i.d.) assumptions on $%
(\epsilon_{i1t},\cdots,\epsilon_{iJt})$ is not needed in the pure
characteristics model.} We therefore obtain the following point
identification result.

\begin{theorem}
\label{thm:fbeta_pure} Suppose Assumptions \ref{as:index}-\ref{as:covariates}%
, and \ref{as:support_pure} hold. Then, $f_\theta$ is identified in the pure
characteristics demand model, where $\theta_{it}=(\beta^{(2)}_{it},%
\alpha_{it}).$
\end{theorem}

\subsection{Bundle choice (Example \protect\ref{ex:bundles})}\label{sec:bundles}

In this section, we consider $\sigma_\epsilon=1$, but it is also possible to analyze the case without the tastes for products.\footnote{For the setting without the tastes for products, we refer to Dunker, Hoderlein, and Kaido (2013), an earlier version of the paper.} We consider an alternative procedure for inverting the demand in Example \ref{ex:bundles}.  This is because this example (and also the example in the next section)
has a specific structure. We note that the inversion of Berry, Gandhi, and
Haile (2013) can still be applied to bundles if one treats each bundle as a
separate good and recast the bundle choice problem into a standard
multinomial choice problem. However, as can be seen from %
\eqref{eq:bundleutil}, Example \ref{ex:bundles} has the additional structure
that the utility of a bundle is the combination of the utilities for each
good and extra utilities, and hence the model does not involve any bundle
specific unobserved characteristic. This structure in turn implies that the
dimension of the unobservable product characteristic $\Xi_t$ equals the
number of goods $J$, while the econometrician in general observes $\dim(S)=%
2^J$ aggregate choice probabilities over bundles, which leads to a system of equations
 whose number of restrictions exceeds the number of unknown quantities. This suggests that (i) using only a part
of the demand system is sufficient for obtaining an inversion, which can be
used to identify $f_\theta$ and (ii) using additional subcomponents of $S$,
one may potentially overidentify the parameter of interest. We therefore
consider an inversion that exploits a monotonicity property of the demand
system that follows from this structure.\footnote{%
The additional structure can potentially be tested. In Example 2, one may
identify the demand for bundles (1,0) and (1,1) using the inversion
described below under the hypothesis that eq. \eqref{eq:bundleutil} holds.
Further, treating (1,0), (0,1), and (1,1) as three separate goods (and (0,0)
as an outside good) and applying the inversion of Berry, Gandhi, and Haile
(2013), one may identify the demand for bundles (1,0) and (1,1) without
imposing \eqref{eq:bundleutil}. The specification can then be tested by
comparing the demand functions obtained from these distinct inversions. We
are indebted to Phil Haile for this point.} For this, we assume that the
following condition is met.

\begin{condition}
\label{cond:bundles} The random coefficient density $f_{\theta}$ is
continuously differentiable. In addition, $(\epsilon_{i1t},\epsilon_{i2t})$ and $%
(D_{1t},D_{2t})$ have full supports in $\mathbb{R}^2$ respectively.
\end{condition}

Let $\tilde{\mathbb{L}}=\{(1,0),(1,1)\}.$ From \eqref{eq:10agdemand}, it is
straightforward to show that $\varphi_{(1,0)}$ is strictly increasing in $%
D_{1t}$ but is strictly decreasing in $D_{2t}$, while $\varphi_{(1,1)}$ is
strictly increasing both in $D_{1t}$ and $D_{2t}$. Hence, the Jacobian
matrix is non-degenerate. Together with a mild support condition on $%
(D_{1t},D_{2t})$, this allows to invert the demand (sub)system and write $%
\Xi_{jt}=\psi_j(X^{(2)}_t,P_t,\tilde S_t)-X^{(1)}_{jt},$ where $\tilde
S_t=(S_{(1,0),t},S_{(1,1),t})$. This ensures Assumption \ref{as:inv1} in
this example (see Lemma \ref{lem:inv1} given in the appendix). By Theorem %
\ref{thm:npiv1}, one can then nonparametrically identify subcomponents $%
(\varphi_{(1,0)},\varphi_{(1,1)})$ of the demand function $\varphi$.

One may alternatively choose $\tilde{\mathbb{L}}=\{(0,0),(0,1)\}$, and the
argument is similar, which then identifies $(\varphi_{(0,0)},%
\varphi_{(0,1)}) $, and hence all components of the demand function $\varphi$
are identified. This inversion is valid even if the two goods are
complements. This is because the inversion uses the monotonicity property of
the aggregate choice probabilities on bundles (e.g. $\phi_{(1,0)}$ and $%
\phi_{(1,1)}$) with respect to $(D_{1t}, D_{2t})$. Hence, even if the
aggregate share of each good (e.g. aggregate share on good 1: $%
\sigma_1=\phi_{(1,0)}+\phi_{(1,1)}$) is not invertible in the price $P_t$
due to the presence of complementary goods, one can still obtain a useful
inversion provided that aggregate choice probabilities on bundles are
observed.

Given the demand for bundles, we now analyze identification of the random
coefficient density. By \eqref{eq:bundleutil}, the demand for bundle (0,0)
is given by 
\begin{align}  \label{eq:dem00}
&\phi_{(0,0)}(x^{(2)},p,\delta) =\int 1\{x^{(2)}_1{}^{\prime
}b^{(2)}+ap_1+e_1<-\delta_1\}1\{x^{(2)}_2{}^{\prime
}b^{(2)}+ap_2+e_2<-\delta_2\}   \\
&\qquad\times1\{(x_1^{(2)}+x_2^{(2)})^{\prime
}b^{(2)}+a(p_1+p_2)+(e_1+e_2)+\Delta<-\delta_1-\delta_2\}f_{%
\theta}(b^{(2)},a,e,\Delta)d\theta.\notag
\end{align}
Given product $j\in\{1,2\}$, let $-j$ denote the other product. We then
define $\tilde\Phi_l$ with $l=(0,0)$ as in the BLP example by letting $%
D_{-jt}$ take a large negative value. For each $(x^{(2)},p,\delta)$, let 
\begin{align}
\tilde\Phi_{(0,0)}(x_{j}^{(2)},p_{j},\delta_{j})\equiv
-\lim_{\delta_{-j}\to-\infty}\phi_{(0,0)}(x^{(2)},p,\delta),~j=1,2.
\label{eq:sumbund}
\end{align}
We then define $\Phi_{(0,0)}$ as in \eqref{eq:Phi}.\footnote{%
In the BLP example, we invert a Radon transform only once. Hence $\Phi$ in %
\eqref{eq:Phi} does not have any subscript. In Examples 2 and 3, we invert
Radon transforms multiple times, and to make this point clear we add
subscripts to $\Phi$ (e.g. $\Phi_{(0,0)}$ and $\Phi_{(1,1)}$).} Consider for
the moment $j=1$ in \eqref{eq:sumbund}. Then, $\Phi_{(0,0)}$ is related to
the joint density $f_{\vartheta_1}$ of $\vartheta_{i1t}\equiv(%
\beta^{(2)}_{it},\alpha_{it},\epsilon_{i1t})$ through a Radon transform.%
\footnote{%
Since the bundle effect $\Delta_{it}$ does not appear in \eqref{eq:dem00},
one may only identify the joint density of the subvector $%
(\beta^{(2)}_{it},\alpha_{it},\epsilon_{i1t})$ from the demand for bundle
(0,0).} Arguing as in \eqref{eq:sum1}, it is straightforward to show that $%
\partial\Phi_{(0,0)}(w,u)/\partial u=\mathcal{R}[f_{\vartheta_1}](w,u)~$
with $w\equiv (x^{(2)}_1,p_1,1)/\|(x^{(2)}_1,p_1,1)\|$ and $u\equiv
\delta_1/\|(x^{(2)}_1,p_1,1)\|$. Hence, one may identify $f_{\vartheta_1}$
by inverting the Radon transform under Assumptions \ref{as:covariates} and %
\ref{as:rc_blp} with $J=2$.

If the researcher is only interested in the distribution of $%
(\beta^{(2)}_{it},\alpha_{it},\epsilon_{ijt})$ but not in the bundle effect,
the demand for $(0,0)$ is enough for recovering their density. However, $%
\Delta_{it}$ is often of primary interest. The demand on (1,1) can be used
to recover its distribution by the following argument. 

The demand for bundle (1,1) is given by 
\begin{align}\label{eq:bundle11}
&\phi_{(1,1)}(x^{(2)},p,\delta) \notag \\
&=\int 1\{x^{(2)}_1{}^{\prime
}b^{(2)}+ap_1+e_1+\Delta>-\delta_1\}1\{x^{(2)}_2{}^{\prime
}b^{(2)}+ap_2+e_2+\Delta>-\delta_2\} \\
&\qquad \times1\{(x_1^{(2)}+x_2^{(2)})^{\prime
}b^{(2)}+a(p_1+p_2)+(e_1+e_2)+\Delta>-\delta_1-\delta_2\}f_{%
\theta}(b^{(2)},a,e,\Delta)d\theta. \notag
\end{align}
Note that $\Delta_{it}$ can be viewed as an additional random coefficient on
the constant whose sign is fixed. Hence, the set of covariates includes a
constant. Again, conditioning on an event where $D_{-jt}$ takes a large
negative value and normalizing the arguments by the norm of $(
x^{(2)}_j,p_j,1)$ yield a function $\Phi_{(1,1)}$ that is related to the
density of $\eta_{ijt}\equiv(\beta^{(2)}_{it},\alpha_{it},\Delta_{it}+%
\epsilon_{ijt})$ through the Radon transform in \eqref{eq:radon}. Note that
the last component of $\eta_j$ and $\vartheta_j$ differ only in the bundle
effect $\Delta_{it}$. Hence, if $\epsilon_{ijt}$ is independent of $%
\Delta_{it}$ conditional on $(\beta^{(2)}_{it},\alpha_{it})$, the
distribution of $\Delta_{it}$ can be identified via deconvolution. For this,
let $\Psi_{\epsilon_{j}|(\beta^{(2)},\alpha)}$ denote the characteristic
function of $\epsilon_{ijt}$ conditional on $(\beta^{(2)}_{it},\alpha_{it})$%
. We summarize these results in the following theorems.

\begin{theorem}
\label{thm:idsub} Suppose Assumptions \ref{as:index}-\ref{as:rc_blp}, \ref%
{as:support_pure} and Condition \ref{cond:bundles} hold with $J=2$ and $%
\theta_{it}=(\beta^{(2)}_{it},\alpha_{it},\Delta_{it},\epsilon_{i1t},%
\epsilon_{i2t})$. Suppose the conditional distribution of $\epsilon_{ijt}$
given $(\beta^{(2)}_{it},\alpha_{it})$ is identical for $j=1,2$.

Then, (a) $f_{\vartheta_j},f_{\eta_j}$ are nonparametrically identified in
Example 2; (b) If, in addition, $\Delta_{it}\perp
\epsilon_{ijt}|(\beta^{(2)}_{it},\alpha_{it})$ and $\Psi_{\epsilon_{j}|(%
\beta^{(2)},\alpha)}(t)\ne 0$ for almost all $t\in\mathbb{R}$ and for some $%
j $, and $\epsilon_{ijt},j=1,2$ are independently distributed (across $j$)
conditional on $(\beta^{(2)}_{it},\alpha_{it})$, then $f_\theta$ is
nonparametrically identified in Example 2.
\end{theorem}

The identification of the distribution of the bundle effect requires the
characteristic function of $\epsilon_{ijt}$ to have isolated zeros (see e.g.
Devroye, 1989, Carrasco and Florens, 2010). This condition can be satisfied
by various distributions including the Type-I extreme value distribution and
normal distribution.

\begin{remark}\rm
\textrm{Note that the conditions of Theorem \ref{thm:idsub} do not impose
any sign restriction on $\Delta_{it}$. Hence, the two goods can be
substitutes $(\Delta_{it}<0)$ for some individuals and complements ($%
\Delta_{it}>0$) for others. This feature, therefore, can be useful for
analyzing bundles of goods whose substitution pattern can significantly
differ across individuals (e.g. E-books and print books). }
\end{remark}

\begin{remark}\rm
\textrm{We note that the utility specification adopted in the pure
characteristics model can also be combined with the bundle choice and
multiple units of consumption studied in Section \ref{sec:mult} in the 
online supplement. The identification of the random
coefficients can be achieved using arguments similar to the ones in Section %
\ref{sec:PCM}%\footnote{The analysis of these settings are contained in an earlier version of this paper, which is available from the authors upon request.}
.}
\end{remark}

\section{Outlook}

This paper is concerned with the nonparametric identification of models of
market demand. It provides a general framework that nests several important
models, including the workhorse BLP model, and provides conditions under
which these models are point identified. Important conclusions include that
the assumption necessary to recover various objects differ; in particular,
it is easier to identify demand elasticities and more difficult to identify
the individual specific random coefficient densities. Moreover, the data
requirements are also shown to vary with the model considered. The
identification analysis is constructive, extends the classical nonparametric
BLP identification as analyzed in BH to other models, and opens up the way
for future research on sample counterpart estimation. A particularly
intriguing part hereby is the estimation of the demand elasticities, as the
moment condition is different from the one used in nonparametric IV.
Understanding the properties of these estimators, and evaluating their
usefulness in an application, is an open research question that we hope this
paper stimulates.

\ifx\undefined\BySame
\fi
\ifx\undefined \let\tmpsmall{\small \renewcommand{\small}{\tmpsmall\sc} \fi

\appendix

\section{Supplement}

In this supplement we consider three extensions to the identification analysis (Section \ref{sec:extensions}), we outline a nonparametric and a parametric estimation procedure (Section \ref{sec:estimation}) and provide the proofs of the theorems in the paper and the supplement (Section \ref{sec:proofs}). The following is a list of notations and definitions used throughout the appendix.

\begin{table}[h]
\begin{center}
\begin{tabular}{rl}
$\mathbb{S}^{q-1}:$ & The unit sphere $\mathbb{S}^{q-1}\equiv%
\{v\in\mathbb{R}^{q}:\|v\|=1\}$. \\ 
$\mathbb{H}_+:$ & The hemisphere $\mathbb{H}_+\equiv\{v=(v_1,v_2,\cdots,
v_{q})\in\mathbb{S}^{q-1}:v_{q}\ge 0\}.$ \\ 
$P_{w,r}:$ & The hyperplane: $P_{w,r}\equiv\{v\in\mathbb{R}%
^{q}:v^{\prime }w=r\}.$ \\ 
$\mu_{w,r}:$ & Lebesgue measure on $P_{w,r}$. \\ 
$\mathcal{R}:$ & Radon transform: $\mathcal{R}[f](w,u)=\int_{P_{w,u}}f(v)d%
\mu _{w,u}(v).$ \\[-24pt]
\end{tabular}
\end{center}
\par
\end{table}

\subsection{Extensions}\label{sec:extensions}
We discuss three extensions to the identification analysis in the paper. The first is a model that accounts for multiple units of consumption in the bundle model. We call this Example 3. The second extension considers alternative-specific coefficients. Finally, we briefly discuss identification of demand $\psi$ with fully independent instruments.

\subsubsection{Multiple units of consumption (Example 3)}\label{sec:mult}

We consider settings where multiple units of consumption are
allowed. For simplicity, we consider the simplest setup where $J=2$ and $%
Y_{1}\in \{0,1,2\}$ and $Y_{2}\in \{0,1\}.$ The utility from consuming $%
y_{1} $ units of product $1$ and $y_{2}$ units of product 2 is specified as
follows: 
\begin{equation}
U_{i,(y_{1},y_{2}),t}^{\ast }=y_{1}U_{i1t}^{\ast }+y_{2}U_{i2t}^{\ast }+
\Delta_{i,(y_{1},y_{2}),t}~,  \label{eq:bundleutil1}
\end{equation}%
where $\Delta _{i,(y_{1},y_{2}),t}$ is the additional utility (or
disutility) from consuming the particular bundle $(y_{1},y_{2})$. This
specification allows, e.g., for decreasing marginal utility (with the number
of units), as well as interaction effects. We assume that $\Delta
_{(1,0)}=\Delta _{(0,1)}=0$ as $U_{i1t}^{\ast }$ and $U_{i2t}^{\ast }$ give
the utility from consuming a single unit of each of the two goods.
Throughout this example, we assume that $U_{i,(y_{1},y_{2}),t}^{\ast }$ is
concave in $(y_{1},y_{2})$. Then, a bundle is chosen if its utility exceeds
those of the neighboring alternatives. For example, bundle $(2,0)$ is chosen
if it is preferred to bundles (1,0), (1,1) and (2,1). That is, 
\begin{align}\label{eq:20demandmult}
&2(X_{1t}^{\prime }\beta_{it} +\alpha_{it} P_{1t}+\Xi
_{1t}+\epsilon_{i1t})+\Delta _{i,(2,0),t}>X_{1t}^{\prime }\beta_{it}
+\alpha_{it} P_{1t}+\Xi _{1t}+\epsilon_{i1t}~,  \notag \\
&2(X_{1t}^{\prime }\beta_{it} +\alpha_{it} P_{1t}+\Xi
_{1t}+\epsilon_{i1t})+\Delta _{i,(2,0),t}  \notag \\
&\hspace{30pt}>X_{1t}^{\prime }\beta_{it} +\alpha_{it} P_{1t}+\Xi
_{1t}+\epsilon_{i1t}+X_{2t}^{\prime }\beta_{it} +\alpha_{it} P_{2t}+\Xi
_{2t}+\epsilon_{i2t}+\Delta _{i,(1,1),t}  \\
&2(X_{1t}^{\prime }\beta_{it} +\alpha_{it} P_{1t}+\Xi
_{1t}+\epsilon_{i1t})+\Delta _{i,(2,0),t}~,  \notag \\
&\hspace{30pt}>2(X_{1t}^{\prime }\beta_{it} +\alpha_{it} P_{1t}+\Xi
_{1t}+\epsilon_{i1t})+X_{2t}^{\prime }\beta_{it} +\alpha_{it} P_{2t}+\Xi
_{2t}+\epsilon_{i2t}+\Delta _{i,(2,1),t}. \notag 
\end{align}%
The aggregate structural demand can be obtained as 
\begin{align}\label{eq:10agdemandmult}
&\varphi _{(2,0)}(X_{t},P_{t},\Xi _{t})=\int 1\{X_{1t}^{\prime
}b+aP_{1t}+e_1+\Delta _{(2,0)}>-\Xi _{1t}\} \notag \\
&\quad \times 1\{(X_{1t}-X_{2t})^{\prime }b+a(P_{1t}-P_{2t})+(e_1-e_2)+\Delta
_{(2,0)}-\Delta _{(1,1)}>-\Xi _{1t}+\Xi _{2t}\} \\
&\quad \times 1\{X_{2t}^{\prime }b+aP_{2t}+e_2+\Delta _{(2,1)}-\Delta _{(2,0)}<-\Xi
_{2t}\}f_{\theta }(b,a,e,\Delta)d\theta ~.\notag 
\end{align}%
The observed aggregate demand for the bundles are defined in a similar way for $%
S_{l,t}=\varphi _{l}(X_{t},P_{t},\Xi _{t})$, $l\in \mathbb{L}$ where $\mathbb{L}%
\equiv \{(0,0),(1,0),(0,1),(1,1),(2,0),(2,1)\}.$

Let $\tilde{\mathbb{L}}=\{(2,0),(2,1)\}.$ From \eqref{eq:20demandmult}, $%
\varphi_{(2,0)}$ is increasing in $D_1$ but is decreasing in $D_2$.
Similarly, $\varphi_{(2,1)}$ is increasing in both $D_1$ and $D_2$. The rest
of the argument is similar to Example 2. This ensures Assumption \ref%
{as:inv1} in this example, and by Theorem \ref{thm:npiv1}, one can then
nonparametrically identify subcomponents $\{\varphi_l,l\in\tilde{\mathbb{L}}%
\}$ of the demand function $\varphi$. One may alternatively take $\tilde{%
\mathbb{L}}=\{(0,0),(0,1)\}$ and use the same line of argument. Note,
however, that (1,0) or (1,1) cannot be included in $\tilde{\mathbb{L}}$ as $%
\phi_{(1,0)}$ and $\phi_{(1,1)}$ are not monotonic in one of $(D_1,D_2)$.
This is because increasing $D_1$ while fixing $D_2$, for example, makes good
1 more attractive and creates both an inflow of individuals who move from
(0,0) to (1,0) and an outflow of individuals who move from (1,0) to (2,0).
Hence, the demand for (1,0) does not necessarily change monotonically.

The nonparametric IV step identifies $\phi_{l}$ for $l\in%
\{(0,0),(0,1),(2,0),(2,1)\}$. Using them, we may first recover the joint
density of some of the random coefficients: $\theta_{it}=(\beta_{it}^{(2)},%
\alpha_{it},\epsilon_{i1t},\epsilon_{i2t},$ $\Delta_{i,(1,1),t},%
\Delta_{i,(2,0),t},\Delta_{i,(2,1),t})^{\prime }.$ We begin with the demand
for $(0,0)$, $(0,1)$, $(2,0)$, and $(2,1)$ given by 
\begin{allowdisplaybreaks}
\begin{align*}
&\phi_{(0,0)}(x^{(2)},p,\delta)=\int 1\{x^{(2)}_1{}^{\prime
}b^{(2)}+ap_1+e_1<-\delta_1\} \\
&\quad\times1\{x^{(2)}_2{}^{\prime }b^{(2)}+ap_2+e_2<-\delta_2\} \\
&\quad\times 1\{(x_{1}^{(2)}+x_{2}^{(2)})^{\prime }b^{(2)}+a
(p_{1}+p_{2})+(e_1+e_2)<-\delta_{1}-\delta_{2}\}
f_{\theta}(b^{(2)},a,e,\Delta)d\theta~, \\
&\phi_{(0,1)}(x^{(2)},p,\delta)=\int 1\{x^{(2)}_2{}^{\prime
}b^{(2)}+ap_2+e_2>-\delta_2\} \\
&\quad\times 1\{(x_{1}^{(2)}-x_{2}^{(2)})^{\prime }b^{(2)}+a
(p_{1}-p_{2})+(e_1-e_2)<-\delta_{1}+\delta_{2}\} \\
&\quad\times 1\{x^{(2)}_1{}^{\prime
}b^{(2)}+ap_1+e_1+\Delta_{(1,1)}>-\delta_1\}f_{\theta}(b^{(2)},a,e,\Delta)d%
\theta~, \\
&\phi_{(2,0)}(x^{(2)},p,\delta)=\int 1\{x^{(2)}_1{}^{\prime
}b^{(2)}+ap_1+e_1+\Delta_{(2,0)}>-\delta_1\} \\
&\quad\times 1\{(x_{1}^{(2)}-x_{2}^{(2)})^{\prime }b^{(2)}+a
(p_{1}-p_{2})+(e_1-e_2)+\Delta_{(2,0)}-\Delta_{(1,1)}>-\delta_{1}+\delta_{2}%
\} \\
&\quad\times 1\{x^{(2)}_2{}^{\prime
}b^{(2)}+ap_2+e_2+\Delta_{(2,1)}-\Delta_{(2,0)}<-\delta_2\}f_{%
\theta}(b^{(2)},a,e,\Delta)d\theta~, \\
&\phi_{(2,1)}(x^{(2)},p,\delta)=\int 1\{x^{(2)}_1{}^{\prime
}b^{(2)}+ap_1+e_1+\Delta_{(2,1)}-\Delta_{(1,1)}>-\delta_1\} \\
&\quad\times 1\{x_{1}^{(2)}{}^{\prime }b^{(2)}+a
p_{1}+e_1+\Delta_{(2,1)}-\Delta_{(2,0)}>-\delta_{2}\} \\
&\quad\times 1\{(x^{(2)}_1+x^{(2)}_2)^{\prime
}b^{(2)}+a(p_1+p_2)+(e_1+e_2)+\Delta_{(2,1)}>-\delta_1-\delta_2\}f_{%
\theta}(b^{(2)},a,e,\Delta)d\theta~.
\end{align*}
Hence, if $D_{2t}$ has a large support, by taking $\delta_2$ sufficiently
small or sufficiently large, we may define 
\begin{align}
\tilde\Phi_{(0,0)}&(x^{(2)}_1,p_1,\delta_1)\equiv-\lim_{\delta_2\to-\infty}%
\phi_{(0,0)}(x^{(2)},p, \delta)  \notag \\
&=-\int 1\{x^{(2)}_1{}^{\prime
}b^{(2)}+ap_1+e_1<-\delta_1\}f_{\theta}(b^{(2)},a,e,\Delta)d\theta~,
\label{eq:mult0} \\
\tilde\Phi_{(0,1)}&(x^{(2)}_1,p_1,\delta_1)\equiv-\lim_{\delta_2\to\infty}%
\phi_{(0,1)}(x^{(2)},p,\delta)  \notag \\
&=-\int 1\{x^{(2)}_1{}^{\prime
}b^{(2)}+ap_1+e_1+\Delta_{(1,1)}>-\delta_1\}f_{\theta}(b^{(2)},a,e,\Delta)d%
\theta~,  \label{eq:mult1} \\
\tilde\Phi_{(2,0)}&(x^{(2)}_1,p_1,\delta_1)\equiv-\lim_{\delta_2\to-\infty}%
\phi_{(2,0)}(x^{(2)},p,\delta)  \notag \\
&=-\int 1\{x^{(2)}_1{}^{\prime
}b^{(2)}+ap_1+e_1+\Delta_{(2,0)}>-\delta_1\}f_{\theta}(b^{(2)},a,e,\Delta)d%
\theta~,  \label{eq:mul2} \\
\tilde\Phi_{(2,1)}&(x^{(2)}_1,p_1,\delta_1)\equiv-\lim_{\delta_2\to\infty}%
\phi_{(2,1)}(x^{(2)},p,\delta)  \notag \\
&=\int 1\{x^{(2)}_1{}^{\prime
}b^{(2)}+ap_1+e_1+\Delta_{(2,1)}-\Delta_{(1,1)}>-\delta_1\}f_{%
\theta}(b^{(2)},a,e,\Delta)d\theta~.  \label{eq:mul3}
\end{align}
For each $l\in\{(0,0),(0,1),(2,0),(2,1)\}$, define $\Phi_{l}$ as in %
\eqref{eq:Phi}. Arguing as in Example 2, $\Phi_l$ is then related to the
random coefficient densities by 
\begin{align*}
\frac{\partial\Phi_{l}(w,u)}{\partial u}&=\mathcal{R}[f_{\vartheta_l}](
w,u),~~~l\in\{(0,0),(0,1),(2,0),(2,1)\},
\end{align*}
\end{allowdisplaybreaks}
where $w\equiv -(x^{(2)}_1,p_1,1)/\|(x^{(2)}_1,p_1,1)\|$ and $%
u\equiv\delta_1/\|(x^{(2)}_1,p_1,1)\|$. Here, for each $l$, $f_{\vartheta_l}$
is the joint density of a subvector $\vartheta_{i,l,t}$ of $\theta_{it}$,
which is given by\footnote{%
Alternative assumptions can be made to identify the joint density of
different components of the random coefficient vector. For example, a large
support assumption on $D_{1t}$ would allow one to recover the joint density
of $(\beta_{it}^{(2)},\alpha_{it},\epsilon_{i2t}+\Delta_{i,(2,1),t}-%
\Delta_{i,(2,0),t})$ from the demand for bundle (2,0).} 
\begin{align}
\begin{split}
\vartheta_{i,(0,0),t}&=(\beta_{it}^{(2)},\alpha_{it},\epsilon_{i1t}),~
\vartheta_{i,(0,1),t}=(\beta_{it}^{(2)},\alpha_{it},\epsilon_{i1t}+%
\Delta_{i,(1,1),t}),~ \\
\vartheta_{i,(2,0),t}&=(\beta_{it}^{(2)},\alpha_{it},\epsilon_{i1t}+%
\Delta_{i,(2,0),t}),~
\vartheta_{i,(2,1),t}\\
&=(\beta_{it}^{(2)},\alpha_{it},\epsilon_{i1t}+%
\Delta_{i,(2,1),t}-\Delta_{i,(1,1),t}).  \label{eq:fbetas}
\end{split}
\end{align}
The joint density of $\theta_{it}$ is identified by making the following
assumption.

\begin{ass}
\label{as:mult_ci} (i) Assume that $(\Delta_{i,(1,1),t},\Delta_{i,(2,0),t},
\Delta_{i,(2,1),t})\perp\epsilon_{ijt}|(\beta^{(2)}_{it},\alpha_{it})$ and $%
\Psi_{\epsilon_{j}|(\beta^{(2)},\alpha)}(t)\ne 0$ for almost all $t\in%
\mathbb{R}$ and for some $j\in\{1,2\}$; (ii) $\epsilon_{ijt},j=1,2$ are
independently and identically distributed (across $j$) conditional on $%
(\beta^{(2)}_{it},\alpha_{it})$; (iii) $(\Delta_{i,(1,1),t},%
\Delta_{i,(2,0),t},\Delta_{i,(2,1),t})$ are independent of each other
conditional on $(\beta^{(2)}_{it},\alpha_{it})$ and $\Psi_{\Delta_{(1,1)}|(%
\beta^{(2)},\alpha)}(t)\ne 0$ for almost all $t\in\mathbb{R}$.
\end{ass}

\textrm{Assumption \ref{as:mult_ci} (iii) means that, relative to the
benchmark utility given as an index function of $(X^{(2)}_t,P_t,D_t)$, the
additional utilities from the bundles are independent of each other.
Assumption \ref{as:mult_ci} (iii) also adds a regularity condition for
recovering the distribution of $\Delta_{i,(2,1),t}$ from those of $%
\Delta_{i,(2,1),t}-\Delta_{i,(1,1),t}$ and $\Delta_{i,(1,1),t}$ through
deconvolution. }

\textrm{Identification of the joint density $f_{\theta}$ allows one to
recover the demand for the middle alternative: (1,0), which remained
unidentified in our analysis in the nonparametric IV step. To see this, we
note that the demand for this bundle is given by 
\begin{multline}
\phi_{(1,0)}(x^{(2)},p,\delta)=\int 1\{0<x_1^{(2)}{}^{\prime
(2)}+ap_1+e_1+\delta_1<-\Delta_{(2,0)}\} \\
\times 1\{x_2^{(2)}{}^{\prime
(2)}+ap_2+e_2+\delta_2<-\Delta_{(1,1)}\}1\{(x_1^{(2)}-x_2^{(2)}){}^{\prime
(2)}+a(p_1-p_2)+(e_1-e_2)<-(\delta_1-\delta_2)\} \\
\times1\{(x_1^{(2)}+x_2^{(2)}){}^\prime
b^{(2)}+a(p_1+p_2)+(e_1+e_2)+\Delta_{(2,1)}<-(\delta_1+\delta_2)\}
f_{\theta}(b^{(2)},a,e,\Delta)d\theta.
\end{multline}
Since the previously unknown density $f_{\theta}$ is identified, this demand
function is identified. This and $\phi_{(1,1)}=1-\sum_{l\in\mathbb{L}%
\setminus \{(1,1)\}}\phi_l$ further imply that all components of $\phi$ are
now identified. We summarize these results below as a theorem.\footnote{%
For simplicity, we only consider the case where $\delta_{2}\to -\infty$ or $%
\infty$ in \eqref{eq:mult0}-\eqref{eq:mult1}. This requires a full support
condition on $D_{1t}$. It is possible to replace this assumption with an
analog of Assumption \ref{as:analytic} by also considering the case where $%
\delta_{1}\to -\infty$ or $\infty$ and imposing an additional restriction on
the distribution of $(\epsilon_{i1t},\epsilon_{i2t},\Delta_{i,(1,1),t},%
\Delta_{i,(2,0),t},\Delta_{i,(2,1),t})$.} }

\begin{theorem}
\label{thm:mult1} Suppose $U_{(y_1,y_2),t}$ is concave in $(y_1,y_2)$. Furthermore, we set $\theta_{it}=\linebreak(\beta^{(2)}_{it},\alpha_{it},\epsilon_{i1t},\epsilon_{i2t},%
\Delta_{i,(1,1),t},\Delta_{i,(2,0),t},\Delta_{i,(2,1),t})$.
Suppose Condition \ref{cond:bundles} and Assumptions \ref{as:index}, \ref%
{as:npiv1}-\ref{as:covariates}, \ref{as:support_pure} hold with $J=2$ and $%
\theta_{it}=(\beta^{(2)}_{it},\alpha_{it},\epsilon_{i1t},\epsilon_{i2t},%
\Delta_{i,(1,1),t},\Delta_{i,(2,0),t},\Delta_{i,(2,1),t})$. Suppose that $%
(X_{1t},P_{1t},D_{1t})$ has a full support. Then, (a) all densities $f_{\vartheta_l}$ for $l\in%
\{(0,0),(0,1),(2,0),(2,1)\}$ are nonparametrically identified in Example 3;
(b) Suppose further that Assumption \ref{as:mult_ci} holds. Then, $%
f_{\theta} $ is identified in Example 3. Further, all components of the
structural demand $\phi$ are identified.
\end{theorem}

\subsubsection{Alternative specific coefficients}

\label{sec:tensor} So far, we have maintained the assumption that $(\beta
_{ijt},\alpha _{ijt})=(\beta_{it} ,\alpha_{it} )$ for all $j$ almost surely.
This excludes alternative specific random coefficients. However, this is not
essential in our analysis. One may allow some or all components of $(\beta
_{ijt},\alpha _{ijt})$ to be different random variables across $j$ and
identify their joint distribution under an extended support condition on the
product characteristics.

We first note that the aggregate demand is identified as long as Assumptions %
\ref{as:index}-\ref{as:npiv1} hold. In the BLP model, the marginal density $%
f_{\vartheta_j}$ of $\vartheta_{ijt}=(\beta^{(2)}_{ijt},\alpha_{ijt},%
\epsilon_{ijt})$ can be identified for any $j$ as long as the corresponding
product characteristics $(X_{jt}^{(2)},P_{jt},D_{jt})$ has a full support
using the same identification strategy in Section \ref{sec:rcid} (see Remark %
\ref{rem:iid}). For the pure characteristics demand model, we note that the
maps $\mathcal{M}_{\mathcal{J}}$ cannot be used because the use of this map is
justified when $(\beta _{ijt},\alpha _{ijt})=(\beta_{it} ,\alpha_{it}
),\forall j$. However, the large support assumption $\mathrm{supp}\,(D_{kt})
= \mathbb{R}$ for $k\ne j$  can still be used to construct $\Phi$. Hence,
the analysis of this case becomes similar to the BLP model. In both models,
the joint density $f_\theta$ of $\theta_{it}=(\vartheta_{i1t},\cdots,%
\vartheta_{iJt})$ can be recovered under the assumption that $\vartheta_{ijt}
$ are independent across $j$.

When the covariates $(X^{(2)}_t,P_t,D_t)$ have rich variations jointly, it
is also possible to identify the joint density $f_\theta$ without the
independence assumption invoked above. This requires us to extend our
identification strategy. To see this, we take Example 2 as an illustration
below. Consider identifying the joint density of $\theta_{it}=(%
\beta^{(2)}_{i1t},\beta^{(2)}_{i2t},\alpha_{i1t},\alpha_{i2t},%
\epsilon_{i1t},\epsilon_{i2t}$ $+\Delta_{it})$ under the assumption that the
two goods are complements, i.e. $\Delta_{it}>0, a.s$. In this setting, we
may use the demand for bundle $(1,0)$, which can be written as 
\begin{align}\label{eq:tensor1}
\begin{split}
\phi_{(1,0)}(x^{(2)},p,\delta)=&\int 1\{x^{(2)}_1{}^{\prime
}b^{(2)}_1+a_1p_1+e_1>-\delta_1\} \\
&\times 1\{x^{(2)}_2{}^{\prime
}b^{(2)}_2+a_2p_2+e_2+\Delta<-\delta_2\}f_{%
\theta}(b^{(2)}_1,b^{(2)}_2,a_1,a_2,\Delta)d\theta.  
\end{split}
\end{align}
To recover the joint density, one has to directly work with this demand
function without simplifying it further. A key feature of \eqref{eq:tensor1}
is that it involves multiple indicator functions and that distinct subsets
of $\theta$ show up in each of these indicator functions. For example, the
first indicator function in \eqref{eq:tensor1} involves $(\beta^{(2)}_{i1t},%
\alpha_{i1t},\epsilon_{i1t})$, while the second indicator function involves $%
(\beta^{(2)}_{i2t},\alpha_{i2t},\epsilon_{i2t}+\Delta_{it})$. Integral
transforms of this form are studied in Dunker, Hoderlein, Kaido, and Sherman (2018)
in their analysis of random coefficients discrete game models. They use
tensor products of integral transforms to study nonparametric identification
of random coefficient densities. Using their framework, one may show that 
\begin{align}
\frac{\partial^2 \phi_{(1,0)}(w_1,w_2,u_1,u_2)}{\partial u_1 \partial u_2}=(%
\mathcal{R}\otimes \mathcal{R})[f_{\theta}](w_1,w_2,u_1,-u_2),
\label{eq:radontensor}
\end{align}
where $w_1=-(x_1^{(2)},p_1,1)/\|(x_1^{(2)},p_1,1)\|$, $%
w_2=(x_2^{(2)},p_2,1)/\|(x_2^{(2)},p_2,1)\|$, and $u_1=\linebreak-\delta_1/%
\|(x_1^{(2)},p_1,1)\|$, $u_2=\delta_2/\|(x_2^{(2)},p_2,1)\|$, and $\mathcal{R%
}\otimes \mathcal{R}$ is the tensor product of Radon transforms, which can
be inverted to identify $f_\theta$. The main principle of our identification
strategy is therefore the same as before. Inverting the transform in %
\eqref{eq:radontensor} to identify $f_\theta$ requires Assumption \ref%
{as:analytic} (i) to be strengthened as follows.

\begin{ass}
\label{as:tensor_support} $( X^{(2)}_{1t}, P_{1t},D_{1t},X^{(2)}_{2t},
P_{2t},D_{2t})$ has a full support.
\end{ass}

This is a stronger support condition than Assumption \ref{as:analytic} (i) as it
requires a joint full support condition for the characteristics of both
goods. This condition is violated, for example, when there is a common
covariate that enters the characteristics of both goods. This is in line
with the previous findings in the literature that identifying the joint
distribution of potentially correlated unobservable tastes for products
(e.g. $\epsilon_{1} $ and $\epsilon_2$) requires variables that are excluded
from one or more goods, see e.g. Keane (1992) and Gentzkow (2007).
Identification of $f_\theta$ is then established by the following theorem.%
\footnote{%
We omit the proof of this result for brevity. Similar to Theorem \ref{thm:fbeta}, it is also possible to establish identification using an analog of Assumption \ref{as:analytic} (ii), which relaxes the support requirement at the cost of an additional moment condition. We also note that one may 
disentangle the distribution of $\Delta_{it}$ from that of $%
\epsilon_{i2t}+\Delta_{it}$ using a deconvolution argument as done in
Theorem \ref{thm:idsub}.}

\begin{theorem}
\label{thm:tensor} In Example 2, let $\theta_{it}=(\beta^{(2)}_{i1t},%
\beta^{(2)}_{i2t},\alpha_{i1t},\alpha_{i2t},\epsilon_{i1t},\epsilon_{i2t}+%
\Delta_{it})$. Suppose that Assumptions \ref{as:index}-\ref{as:npiv1}, \ref%
{as:covariates}, and \ref{as:tensor_support} hold. Suppose further that $%
\Delta_{it}>0,~a.s.$ Then, $f_\theta$ is identified.
\end{theorem}

\subsection{Nonparametric identification of $\protect\psi$ with full
independence}\label{sec:fullind} 

In Section \ref{sec:inversion}, we discussed the 
nonparametric identification of the structural functions $\psi_j$ in the equation $%
\Xi_{jt}=\psi_j(X_t^{(2)},P_t,\tilde S_t)-X^{(1)}_{jt}$. Following BH
(2013), we proposed to identify the structural functions by the conditional
moment equations 
\begin{align*}
E\Big[\psi_j\left(X_t^{(2)},P_t,S_t\right)\Big|Z_t=z_t,X_t=\left(x_t^{(1)},
x_t^{(2)}\right)\Big]= x_{jt}^{(1)},~~j=1,\cdots,J.
\end{align*}
with instrumental variables $Z_t$. The identification relies on the
assumption that the unobservable $\Xi_{jt}$ is mean independent of the
instruments. However, in many applications researchers choose instruments by
arguing that they are independent of the unobservable. Using only mean
independence means using only parts of the available information. Thereby,
the identifying power is weakened. Adding the stronger independence
assumption when it is justified will improve identification as well as
estimation. Therefore, we propose an approach similar to Dunker et al.
(2014) by formally assuming 
\begin{equation*}
\Xi_{jt} \mathpalette{\protect\independenT}{\perp} (Z_t, X_t) \quad %
\mbox{and } E[\Xi_{jt}] = 0 \qquad \mbox{for all } j,t.
\end{equation*}
This leads to the nonlinear equation 
\begin{equation*}
0 \!=\! \left(\!\!%
\begin{array}{c}
P[\psi_j(X_t^{(2)}, S_t,P_t) - X_{jt}^{(1)} \le \xi] - P[\psi_j(X_t^{(2)},
S_t, P_t) - X_{jt}^{(1)} \le \xi| Z_t=z_t, X_t=x_t] \\ 
E[\psi_j(X_t^{(2)}, S_t,P_t) - X_{jt}^{(1)}]%
\end{array}%
\!\right)
\end{equation*}
for all $\xi,z_t,x_t$. Nonparametric estimation of problems involving this
type of nonlinear restrictions are studied in 
Dunker et al. (2014). To give sufficient conditions for identification, we
define the operator 
\begin{align*}
&F\left(\varphi\right)(\xi,z_t,x_t) :=\\
&\left(%
\begin{array}{c}
P[\varphi(X_t^{(2)}, S_t,P_t) - X_{jt}^{(1)} \le \xi] - P[\varphi(X_t^{(2)},
S_t, P_t) - X_{jt}^{(1)} \le \xi| Z_t=z_t, X_t=x_t] \\ 
E[\varphi(X_t^{(2)}, S_t,P_t) - X_{jt}^{(1)}]%
\end{array}%
\right).
\end{align*}
The function $\psi_j$ is a root of the operator $F$. It is, therefore,
globally identified under the following assumption.

\begin{ass}
\label{ass:global} The operator $F$ has a unique root.
\end{ass}

On first sight this may appear as a strong assumption due to the complexity
of the operator. It is, however, weaker than the usual completeness
assumption for the mean independence assumption. This is because, if $\Xi
_{jt}\mathpalette{\protect\independenT}{\perp}(Z_{t},X_{t})$ and the usual
completeness assumption hold, then $F$ has only one root. On the other hand,
completeness is not necessary for $F$ to have a unique root. Hence, when $%
\Xi _{jt}\mathpalette{\protect\independenT}{\perp}(Z_{t},X_{t})$, Assumption %
\ref{ass:global} is weaker than Assumption \ref{as:npiv1}. Another important
advantage of this method is that because the $D_{j}$ do not vanish, we have
a close analog to nonparametric IV with full independence, see, Dunker et
al. (2014) and Dunker (2021), where $D_{j}$ now plays the role of the
dependent variable.

\subsection{Suggested estimation methods}\label{sec:estimation}

\subsubsection{Nonparametric estimator}

The structure of the nonparametric identification suggests a nonparametric
estimation strategy in a natural way. It consists of three steps. The first
step is the estimation of the structural function $\psi_j$. The second step
is to derive the function $\Phi$ from the estimated $\widehat \psi_j$. The
last step of the estimation is the inversion of a Radon transform.

The mathematical structure of the first step is similar to nonparametric IV.
The conditional expectation operator on the left hand side of the equation 
\begin{equation*}
E[\psi
_{j}(x_{t}^{(2)},P_{t},S_{t})|Z_{t}=z_{t},X_{t}=x_{t}]=x_{jt}^{(1)}\qquad 
\mbox{for
all }x_{t},z_{t}
\end{equation*}%
has to be inverted. Let us denote this linear operator by $T$ and rewrite
the problem as $(T\psi _{j})(z_{t},x_{t})=x_{jt}^{(1)}$. Here $x_{jt}^{(1)}$
should be interpreted as a function in $x_{t}$ and $z_{t}$ which is constant
in $x_{t}^{(2)}$, $z_{t}$, and $x_{it}^{(1)}$ for $i\neq j$. The operator
depends on the joint density of $(X_{t},P_{t},S_{t},Z_{t})$ which has to be
estimated nonparametrically, e.g. by kernel density estimation. This gives
an estimator $\widehat{T}$. As in nonparametric IV the operator equation is
usually ill-posed, and regularized inversion schemes must be applied. We
propose Tikhonov regularization for this purpose: 
\begin{equation}
\widehat{\psi }_{j}:=\min_{\psi }\Vert \widehat{T}\psi -x_{jt}^{(1)}\Vert
_{L^{2}(X_{t},Z_{t})}^{2}+\alpha \mathfrak{R}(\psi ).\label{eq:tikhonov}
\end{equation}%
Here, $\alpha \geq 0$ is a regularization parameter and $\mathfrak{R}$ a
regularization functional. A common choice is $\mathfrak{R}(\psi )=\Vert
\psi \Vert _{L^{2}}^{2}$, however, if more smoothness is expected, this can
be a squared Sobolev norm or some other norm. In the case of bundles and
multiple goods we know that $\psi $ must be monotonically increasing or
decreasing in $S_{t}$. One may incorporate this a priori knowledge by
setting $\mathfrak{R}(\psi )=\infty $ for all functions $\psi $ not having
this property. Since monotonicity is a convex constraint, even with this
choice of $\mathfrak{R}$, equation \eqref{eq:tikhonov} is a convex
minimization problem. Solving the problem is computationally feasible, see
Eggermont (1993), Burger and Osher (2004), and Resmerita (2005) for
regularization with general convex regularization functional. Furthermore,
we refer to Newey and Powell (2003) for the related nonparametric IV problem.

In the second step $\widehat{\psi}_j(X_t^{(2)},P_t,S_t)$ is inverted in $S_t$
to get an estimate $\widehat{\phi}_j$ for the demand function $\phi_j$. In
the BLP model, we approximate the limit of $\widehat\psi_j(X_t^{(2)},P_t,S_t)
$ for $D_{kt} \rightarrow -\infty$ to construct an estimate for $\tilde
\Phi_j$ as in \eqref{eq:blp_phitilde}. When $\epsilon_{ijt}$ is iid across $j
$, one may improve efficiency by repeating this process for all products and
averaging $\tilde \Phi_j$ across $j=1,\cdots,J$. For the pure
characteristics model an estimate of $\tilde \Phi_j$ is computed from $%
\widehat{\phi}_j$ by a sum over permutations as in \eqref{eq:sum_pure}.
Similar constructions can be carried out for the models of bundel choices %
\eqref{eq:sumbund} and multiple unites of consumption \eqref{eq:mult0} -- %
\eqref{eq:mul3}. From an estimator of $\tilde \Phi_j$ we get an estimate $%
\widehat{\Phi}$ of $\Phi$ by normalization as in \eqref{eq:Phi} or %
\eqref{eq:Phi_pure}.

The third step of our nonparametric estimation strategy is the inversion of
a Radon transform.
A popular and efficient method for the problem is the filtered back
projection 
\begin{equation*}
\widehat f_\theta(\vartheta) = \mathcal{R}^*\left( \Omega_r *_\delta \frac{%
\partial \Phi _{j}(x_{j}^{(2)},p_{j},\delta _{j})}{\partial \delta _{j}}
\right)(\vartheta).
\end{equation*}
Here $\vartheta = (b,a,e)$ in the BLP model, $\vartheta = (b,a)$ in the PCM,
or $\vartheta = (b,a,\Delta)$ in other models. The operator $(R^*g)(x):=
\int_{\|w\| = 1} g(w,w^{\prime }x)dw$ is the adjoint of the Radon transform,
and $*_\delta$ denotes the convolution with respect to the last variable $%
\delta_j$, and $\Omega_r$ is the function 
\begin{equation*}
\Omega_r(s) := \frac{1}{4\pi^2} 
\begin{cases}
(\cos(rs)-1)/s^2 \qquad & \mbox{for } s \neq 0, \\ 
r^2/2 \qquad & \mbox{for } s = 0.%
\end{cases}%
\end{equation*}
For more details on this algorithm in a deterministic setting we refer to
Natterer (2001). Alternative estimator for random coefficients are proposed and
analyzed in Hoderlein, Klemel\"a, and Mammen (2010) and Dunker, Mendoza, and Reale (2021).

\subsubsection{Parametric estimators for bundle choice models}

Our nonparametric identification analysis shows that the choice of bundles
and multiple units of consumption can be studied very much in the same way
as the standard BLP model (or the pure characteristic model). This suggests
that one may construct parametric estimators for these models by extending
standard estimation methods, given appropriate data. Below, we take Example
2 and illustrate this idea.

Let $\theta_{it}=(\beta^{(2)}_{it},\alpha_{it},\Delta_{it},\epsilon_{1it},%
\epsilon_{2it})$ be random coefficients and let $f_\theta(\cdot;\gamma)$ be
a parametric density function, where $\gamma$ belongs to a finite
dimensional parameter space $\Gamma\subset\mathbb{R}^{d_\gamma}$. The
estimation procedure consists of the following steps:

\begin{description}
\item[Step 1] : Compute the aggregate share of bundles as a function of
parameter $\gamma$ conditional on the set of covariates.

\item[Step 2] : Use numerical methods to solve demand systems for $%
(D_{1t},D_{2t})$, where $D_{jt}=\Xi_{jt}+X^{(1)}_{jt},j=1,2$ and obtain the
inversion in eq. \eqref{eq:invxi}.

\item[Step 3] : Form a GMM criterion function using instruments and minimize
it with respect to $\gamma$ over the parameter space.
\end{description}

The first step is to compute the aggregate share. One may approximate the
aggregate share of each bundle such as the one in \eqref{eq:10agdemand} by
simulating $\theta $ from $f_{\theta }(\cdot ;\gamma )$ for each $\gamma .$
Specifically, if the conditional CDF of $\epsilon_{ijt}$ given $(\beta
^{(2)}_{it},\alpha_{it},\Delta_{it})$ has an analytic form, the two-step
method in BLP and Berry and Pakes (2007) can be employed. We take the demand
for bundle (0,0) in eq. \eqref{eq:dem00} as an example. Conditional on the
product characteristics $y\equiv (x^{(2)},p,\delta )$ and the rest of the
random coefficients $(\beta ^{(2)}_{it},\alpha_{it},\Delta_{it} )$, bundle
(0,0) is chosen when 
\begin{align}
\sigma_\epsilon\epsilon_{i1t}<h_1(y,b ^{(2)},a,\Delta )& ~~\text{ and }%
~~\sigma_\epsilon\epsilon_{i2t}<h_2(y,b ^{(2)},a,\Delta ),~\text{ if }%
~\Delta<0  \label{eq:alphabds} \\
\sigma_\epsilon\epsilon_{i1t}<h_2(y,b ^{(2)},a,\Delta )& ~~\text{ and }%
~~\sigma_\epsilon(\epsilon_{i1t}+\epsilon_{i2t})<h_3(y,b ^{(2)},a,\Delta ),~%
\text{ if }~\Delta\ge 0,  \label{eq:alphabds1}
\end{align}%
where 
\begin{multline}
h_1(y,\beta ^{(2)},a,\Delta )\equiv -x_{1}^{(2)}{}^{\prime}
b^{(2)}-ap_1-\delta _{1},~~ h_2(y,\beta ^{(2)},a,\Delta ) \equiv
-x_{2}^{(2)}{}^{\prime }b^{(2)}-ap_2-\delta _{2}, \\
h_3(y,\beta^{(2)},a,\Delta)\equiv -(x_{1}^{(2)}+x_{2}^{(2)})^{\prime
(2)}-a(p_1+p_2)-(\delta_1-\delta_2).
\end{multline}%
In what follows, we consider the BLP setting where $\sigma_\epsilon=1$.%
\footnote{%
In the PCM, one may adopt a similar approach by letting one of the remaining
random coefficients play the role of $\epsilon_{ijt}$. For example, replace %
\eqref{eq:alphabds}-\eqref{eq:alphabds1} with 
\begin{align}
a<h_1(y,b^{(2)},\Delta),&\text{ and } a<h_2(y,b^{(2)},\Delta),\text{ if }%
\Delta<0 \\
a<h_2(y,b^{(2)},\Delta),&\text{ and } a<h_3(y,b^{(2)},\Delta),\text{ if }%
\Delta\ge0,
\end{align}
where $h_1(y,b^{(2)},\Delta)=(-x_1^{(2)}{}^{\prime (2)}-\delta_1)/p_1$, and $%
h_2,h_3$ are defined similarly. Specify the conditional distribution of $%
\alpha_{it}$ so that an analog of \eqref{eq:integrand} can be calculated.
The rest of the estimation procedure is similar.} Specify the conditional
distribution of $(\epsilon_{i1t},\epsilon_{i2t})$ given $(\beta^{(2)}_{it},%
\alpha_{it},\Delta_{it})$. For each $(y,b^{(2)},a,\Delta)$, define 
\begin{align}\label{eq:integrand}
\begin{split}
&G(y,b^{(2)},a,\Delta)\equiv\\
&\;\;\begin{cases}
Pr(\epsilon_{i1t}<h_1(y,b ^{(2)},a,\Delta ),~\epsilon_{i2t}<h_2(y,b
^{(2)},a,\Delta )|y,b^{(2)},a,\Delta) & \Delta<0 \\ 
Pr(\epsilon_{i1t}<h_2(y,b ^{(2)},a,\Delta
),~\epsilon_{i1t}+\epsilon_{i2t}<h_3(y,b ^{(2)},a,\Delta
)|y,b^{(2)},a,\Delta) & \Delta>0.%
\end{cases}
\end{split}
\end{align}
The value of $G(y,b^{(2)},a,\Delta)$ can be calculated analytically, for
example, if one specifies the joint distribution of $(\epsilon_{i1t},%
\epsilon_{i2t})$ as normal. Eq. \eqref{eq:alphabds}-\eqref{eq:alphabds1}
then imply that the aggregate share of bundle (0,0) is given by 
\begin{align}
\phi _{(0,0)}(x^{(2)},p,\delta ;\gamma )=\int G(y,b^{(2)},a,\Delta) f_{\beta
^{(2)},a,\Delta }(b,a,\Delta ;\gamma )d\theta .
\end{align}%
This can be approximated by the simulated moment: 
\begin{align}\label{eq:share}
\hat{\phi}_{(0,0)}(x^{(2)},p,\delta ;\gamma )=\frac{1}{n_{S}}%
\sum_{i=1}^{n_{S}}G\left(y,b_{i}^{(2)},a_i,\Delta _{i}\right), 
\end{align}%
where the sample $\{(b_{i}^{(2)},a_i,\Delta _{i}),i=1,\cdots
,n_{S}\}$ is generated from $f_{\beta ^{(2)},a,\Delta }(\cdot ;\gamma )$.%
\footnote{%
One may also use an importance sampling method.} Computation of the
aggregate demand for other bundles is similar. This step therefore gives the
model predicted aggregate demand $\hat{\phi}_{l}$ for all bundles under a
chosen parameter value $\gamma $.

The next step is then to invert subsystems of demand and obtain $\psi$
numerically. Given $\hat\phi_l,l\in\mathbb{L}$ from Step 1, this step can be
carried out by numerically calculating inverse mappings. For example, take $%
\tilde{\mathbb{L}}=\{(0,0),(0,1)\}$. Then, $(\delta_1,\delta_2)\mapsto
\big(\hat\phi_{(0,0)}(x^{(2)},p,\delta;\gamma),$ $\hat\phi_{(0,1)}(x^{(2)},p,%
\delta;\gamma)\big)$ defines a mapping from $\mathbb{R}^2$ to $[0,1]^2$.
Standard numerical methods such as the Newton-Raphson method or the homotopy
method (see Berry and Pakes, 2007) can then be employed to calculate the
inverse of this mapping\footnote{%
Whether the demand subsystems admit an analog of BLP's contraction mapping
method is an interesting open question, which we leave for future research.}%
, which then yields $\hat\psi(\cdot;\gamma)\equiv(\hat
\psi_{1}(\cdot;\gamma),\hat\psi_{2}(\cdot;\gamma))$ such that 
\begin{align}\label{eq:oid1}
\begin{split}
\Xi_{1,t}&=\hat\psi_1(X^{(2)}_t,P_t,S_{(0,0),t},S_{(0,1),t};%
\gamma)-X^{(1)}_{1t}\\
\Xi_{2,t}&=\hat\psi_2(X^{(2)}_t,P_t,S_{(0,0),t},S_{(0,1),t};\gamma)-X^{(1)}_{2t}
\end{split}
\end{align}
where $(S_{(0,0),t},S_{(0,1),t})$ are observed shares of bundles. One may
further repeat this step with $\tilde{\mathbb{L}}=\{(1,0),(1,1)\}$, which
yields 
\begin{align}\label{eq:oid2}
\begin{split}
\Xi_{1,t}&=\hat\psi_3(X^{(2)}_t,P_t,S_{(1,0),t},S_{(1,1),t};%
\gamma)-X^{(1)}_{1t}\\
\Xi_{2,t}&=\hat\psi_4(X^{(2)}_t,P_t,S_{(1,0),t},S_{(1,1),t};\gamma)-X^{(1)}_{2t}
\end{split}
\end{align}
This helps to generate additional moment restrictions in the next step.

The third step is to use \eqref{eq:oid1}-\eqref{eq:oid2} to generate moment
conditions and estimate $\gamma$ by GMM. There are four equations in total,
while because the shares sum up to 1 one equation is redundant. Hence, by
multiplying instruments to the residuals from the first three equations, we
define the sample moment: 
\begin{align*}
g_n(X_t,P_t,S_t,Z_t;\gamma)\equiv\frac{1}{n}\sum_{t=1}^n%
\begin{pmatrix}
\hat\psi_1(X^{(2)}_t,P_t,S_{(0,0),t},S_{(0,1),t};\gamma)-X^{(1)}_{1t} \\ 
\hat\psi_2(X^{(2)}_t,P_t,S_{(0,0),t},S_{(0,1),t};\gamma)-X^{(1)}_{2t} \\ 
\hat\psi_3(X^{(2)}_t,P_t,S_{(1,0),t},S_{(1,1),t};\gamma)-X^{(1)}_{1t}%
\end{pmatrix}
\otimes 
\begin{pmatrix}
Z_t \\ 
X_t%
\end{pmatrix}%
.
\end{align*}
Letting $W_n(\gamma)$ be a (possibly data dependent) positive definite
matrix, define the GMM criterion function by 
\begin{align*}
Q_n(\gamma)\equiv g_n(X_t,P_t,S_t,Z_t;\gamma)^{\prime }W_n(\gamma)
g_n(X_t,P_t,S_t,Z_t;\gamma).
\end{align*}
The GMM estimator $\hat\gamma$ of $\gamma$ can then be computed by
minimizing $Q_n$ over the parameter space. A key feature of this method is
that it uses the familiar BLP methodology (simulation, inversion \& GMM) but
yet allows one to estimate models that do not fall in the class of
multinomial choice models. Employing our procedure may, for example, allow
one to estimate bundle choices (e.g. print newspaper, online newspaper, or
both) or platform choices using market level data.

\subsection{Proofs}\label{sec:proofs}

\begin{proof}[\rm \bf Proof of Theorem \ref{thm:npiv1}]
The proof of the theorem is immediate from Theorem 1 in BH (2014). We therefore give a brief sketch.
By Assumptions \ref{as:index} and \ref{as:inv1}, we note that there exists a function $\psi:\mathbb R^{Jk_2}\times\mathbb R^J\times\mathbb R^J\to\mathbb R^J$ such that for some subvector $\tilde S_t$ of $S_t$,
\begin{align*}
\Xi_{jt}=\psi_j(X_t^{(2)},P_t,\tilde S_t)-X_{jt}^{(1)}~,~j=1,\cdots,J,
\end{align*}
and by Assumption \ref{as:npiv1}, the following moment condition holds:
\begin{align*}
	E[\psi_j(X_t^{(2)},P_t,\tilde S_t)-X_{jt}^{(1)}|Z_t,X_t]=0~.
\end{align*}
Identification of $\psi$ then follows from applying the completeness argument in the proof of Theorem 1 in BH (2014).
\end{proof}

\begin{lemma}\label{lem:limited_angle}
Let $\theta = (\theta_1,\ldots;\theta_d)$ be a $d$-dimensional random vector with density $f_\theta$. Assume that the moments of all components are finite $\E[|\theta_d|^l] < \infty$ for all $i = 1,\ldots, d$ and $l = \mathbb{N}$. In addition, let for any $z > 0$
\[
\lim_{p \rightarrow \infty} \frac{z^l}{l!} E\left[\left(|\theta_1| + |\theta_2| + \ldots + |\theta_d|\right)^l\right]=0.
\]
For any open neighborhood $\mathcal{U} \subset \mathbb{S}^{d-1}$ it holds that if the Radon transform of $\mathcal{R}[f_\theta](w,\delta)$ is known for all $(w,\delta) \in \big\{(w,w't)|w\in\mathcal{U}, t \in \supp(\theta)\big\}$, the density $f_\theta$ is identified.
\end{lemma}

\begin{proof}[\rm \bf Proof of Lemma \ref{lem:limited_angle}]
We first show that $\mathcal{F}f_\theta$ the Fourier transform of $f_\theta$ is analytic. The Fourier transform can be approximated by the $p$-th Taylor polynomial for some point $b_0 \in \R^d$. The Taylor remainder for some point $b \in \R^d$ is bounded by
\[
R_p (\mathcal{F}f_\theta)(b;b_0) \le \sum_{\alpha \in \mathbb{N}^d, |\alpha|=p+1} \frac{(b-b_0)^\alpha}{\alpha!} \left\|D^\alpha \mathcal{F}f_\theta \right\|_\infty.
\]
In this formula the multi-index notation is used with respect to $\alpha$. This means 
\newline$\alpha = (\alpha_1, \alpha_2, \ldots, \alpha_d) \in \mathbb{N}^d$, $|\alpha| := \sum_{i = 1}^d \alpha_i$, $\alpha! := \prod_ {i = 1}^d \alpha_i!$, and
\[
D^\alpha\mathcal{F}f_\theta = \frac{\partial^{|\alpha|}\mathcal{F}f_\theta}{\partial b_1^{\alpha_1} \partial b_2^{\alpha_2} \ldots \partial b_k^{\alpha_k} }. 
\]
Note that
\begin{align*}
\left\|D^\alpha \mathcal{F}f_\theta\right\|_\infty &\le \int_{\R^d}|v_1^{\alpha_1} v_2^{\alpha_2} \ldots v_d^{\alpha_d}| f_\theta(v_1,v_2,\ldots,v_d)dv\\
& \le \int_{\R^d}|v_1|^{\alpha_1} |v_2|^{\alpha_2} \ldots |v_d|^{\alpha_d} f_\theta(v_1,v_2,\ldots,v_d)dv\\
&= E [|\theta_2|^{\alpha_1} |\theta_2|^{\alpha_2} \ldots |\theta_d|^{\alpha_k}|].
\end{align*}
This yields
\begin{align*}
R_p (\mathcal{F}f_\theta)(b;b_0) &\le \|b-b_0\|_\infty^p E\left[ \sum_{\alpha \in \mathbb{N}^d, |\alpha|=p+1} \frac{|\theta_1|^{\alpha_1} |\theta_2|^{\alpha_2} \ldots |\theta_d|^{\alpha_d}}{\alpha!}\right]\\
& \le \|b-b_0\|_\infty^p E\left[ (p!)^{-1}\left(|\theta_1| + |\theta_2| + \ldots + |\theta_d|\right)^p \right]\\
& \le \frac{\|b-b_0\|_\infty^p}{p!} E\left[ \left(|\theta_1| + |\theta_2| + \ldots + |\theta_d|\right)^p \right].
\end{align*}
Hence, the Taylor approximation converges point-wise to $\mathcal{F}f_\theta$ on $\R^d$. Consequently, if $\mathcal{F}f_\theta$ is know on some neighborhood around $b_0$, $\mathcal{F}f_\theta$ is identified. This makes $\mathcal{F}f_\theta$ an analytic function. Since the Fourier transform is bijective, this identifies $f_\theta$ as well.

It remains to show that $\mathcal{F}f_\theta$ is known in some open neighborhood. By the Fourier slice theorem for the Radon transform $(\mathcal{F}f_\theta)(w\eta) = \mathcal{F}_1(\mathcal{R}f_\theta[w,\cdot])(\eta)$. Here $\mathcal{F}_1$ denotes the one-dimensional Fourier transform that acts on the free variable denoted by ``$~\cdot~$''. Note that $\mathcal{R}f_\theta[w,\delta] = 0$ if $w \in \mathcal{U}$ but $(w,\delta) \notin \big\{(w,w't)|w\in\mathcal{U}, t \in \supp(\theta)\big\}$. Thus, if $\mathcal{R}f_\theta[w,\delta]$ is known for all $(w,\delta) \in \big\{(w,w't)|w\in\mathcal{U}, t \in \supp(\theta)\big\}$, it is known for all
$w \in \mathcal{U}$ and all $\delta \in \R$. It follows that $\mathcal{F}f_\theta$ is known on some open neighborhood. This identifies $f_\theta$.
\end{proof}

\begin{proof}[\rm \bf Proof of Theorem \ref{thm:fbeta}]
(i) First, under the linear random coefficient specification, the connected substitutes assumption in Berry, Gandhi, and Haile (2013) is satisfied. By Theorem 1 in Berry, Gandhi, and Haile (2013), Assumption \ref{as:inv1} is satisfied. Then,
by Assumptions \ref{as:index}-\ref{as:npiv1} and Theorem \ref{thm:npiv1}, $\psi$ is identified. Further, the aggregate demand $\phi$ is identified by \eqref{eq:blpeq} and the identity $\phi_0=1-\sum_{j=1}^J\phi_j$. 

For any product $j$ and product characteristics $(x_j^{(2)},p_j,\delta_j)$ define the new function
\[
\tilde \Phi_j(x_j^{(2)},p_j,\delta_j) = - \lim_{\delta_1,\ldots,\delta_{j-1},\delta_{j+1},\ldots, \delta_J \rightarrow -\infty} \phi_j(x^{(2)},p,\delta)
\]
point wise. Here $\phi_j(x^{(2)},p,\delta)$ can be any fixed vector of product characteristics where $(x_j^{(2)},p_j,\delta_j)$ coincide with the values on the l.h.s. of the equation. The limit on the r.h.s. exists and is unique. This can be seen by using the definition of $\phi_j$, Lebesgue's theorem, and Assumption \ref{as:rc_blp}. Consequently,
\[
\tilde \Phi(x_j^{(2)},p_j,\delta_j) = -\int 1\{x_j^{(2)}{}'b^{(2)}+ap_j + \epsilon_j < -\delta_j\}f_{\vartheta_j}(b^{(2)},a,e_j)d\vartheta_j.
\]

Now define $\Phi$ as in \eqref{eq:Phi} and conclude
\begin{multline}
\Phi(w,u )=-\int 1\{w'\theta<-u \}f_{\vartheta_j
}(b^{(2)},a,e_j)d\theta \\
=-\int_{-\infty }^{-u}\int_{P_{w,r}}f_{\vartheta_j}(b^{(2)},a,e_j)d\mu _{{%
w,r}}(b^{(2)},a,e_j)dr=-\int_{-\infty }^{-u}\mathcal{R}[f_{\vartheta_j
}](w,r)dr~.\label{eq:radonap}
\end{multline}%
Taking a derivative with respect to $u$ yields \eqref{eq:radon1}.
By the assumption that the conditional distribution of $\epsilon_{ijt}$ given $(\beta^{(2)}_{it},\alpha_{it})$ is identical for $j=1,\cdots,J$, it follows that $f_{\vartheta_j}=f_{\vartheta},\forall j$ for some common density $f_{\vartheta}$. Hence, we may rewrite  \eqref{eq:radon1} as
\begin{equation}
\frac{\partial \Phi(w,u)}{\partial u}=\mathcal{R}[f_{\vartheta }](w,u).
\end{equation}
Note that by Assumptions \ref{as:covariates} (i) and \ref{as:rc_blp}, 
$\partial\Phi(w,u)/\partial u$ is well-defined for some $(w,u)\in\mathbb H_+\times\mathbb R$. 
By Assumption \ref{as:analytic} $\partial\Phi(w,u)/\partial u$ is either identified for all $(w,u)\in\mathbb H_+\times\mathbb R$ or only for $w$ in some open neighborhood of $\mathbb H_+$.
In the first case the identification of $f_{\vartheta}$ follows from the injectivity of the Radon transform (Theorem I in Cram\'er and Wold, 1936). In the second case the the identification of $f_{\vartheta}$ follows from Lemma \ref{lem:limited_angle}.

(ii) In the first part of the proof $f_{\vartheta_j}$, $j=1,2,\ldots,J$ were identified (as $f_\vartheta$). Hence, the conditional distribution $f_{\epsilon_{j}|\beta^{(2)},\alpha}$ of $\epsilon_{ijt}$ given $(\beta^{(2)}_{it},\alpha_{it})$ and the marginal distribution $f_{\beta^{(2)},\alpha}$ of $(\beta^{(2)}_{it},\alpha_{it})$ are identified for any $j$.
Under the additional assumption that $\epsilon_{i1t}, \epsilon_{i2t}, \ldots, \epsilon_{iJt}$ are independent conditional on $(\beta^{(2)}_{it},\alpha_{it})$, we get the joint distribution of $\theta_{it}$ by
\begin{align}
f_{\theta}(b^{(2)},\alpha,e_1, \ldots, e_J) = \prod_{j=1}^J {f_{\epsilon_{j}|\beta^{(2)},\alpha}} (e_j|b^{(2)},\alpha)\times f_{\beta^{(2)},\alpha}(b^{(2)},\alpha).	\label{eq:ftheta_id}
\end{align}
Hence, $f_\theta$ is identified.
\end{proof}

\begin{proof}[\rm \bf Proof of Theorem \ref{thm:fbeta_pure}]
	First, under the linear random coefficient specification, the connected substitutes assumption in Berry, Gandhi, and Haile (2013) is satisfied. By Theorem 1 in Berry, Gandhi, and Haile (2013), Assumption \ref{as:inv1} is satisfied. Then,
by Assumptions \ref{as:index}-\ref{as:npiv1} and Theorem \ref{thm:npiv1}, $\psi$ is identified. Further, the aggregate demand $\phi$ is identified by \eqref{eq:blpeq} and the identity $\phi_0=1-\sum_{j=1}^J\phi_j$. By Assumption \ref{as:support_pure}, for each  $(\mathbf x^{(2)},\mathbf p,\mathbf d)\in \mathbb R^{d_X-1}\times\mathbb R\times \mathbb R$, there is a product (say $j$), with respect to which the marginalization of the demand is permitted. Therefore, there is $(x^{(2)},p,\delta)\in\mathcal H_j$ whose coordinate projection is $(\mathbf x^{(2)},\mathbf p,\mathbf d)$. Hence, one may construct
\begin{align}\label{eq:MJsum}
\begin{split}
\tilde \Phi_j(\mathbf x^{(2)},\mathbf p,\mathbf d)&=\sum_{\mathcal J\subseteq\{1,\cdots,J\}\setminus\{j\}}  \phi_j\circ \mathcal M_{\mathcal J}(x^{(2)},p,\delta)\\
&=\int 1\{x_j^{(2)}{}'b^{(2)}+ap_j<-\delta_j\}f_\theta(b^{(2)},a)d\theta,
\end{split}
\end{align}
where the second equality follows because of the following. First, $\mathcal M_{\mathcal J}$ replaces the indicators in $\phi_j$ of the form $1\{(x_j^{(2)}-x_i^{(2)})'b^{(2)}+a(p_j-p_i)<-(\delta_j-\delta_i)\}$ with $1\{(x_j^{(2)}-x_i^{(2)})'b^{(2)}+a(p_j-p_i)>-(\delta_j-\delta_i)\}$ for $i\in\mathcal J$. The random coefficients are assumed to be  continuously distributed. We therefore have
\begin{multline*}
	1\{(x_j^{(2)}-x_i^{(2)})'b^{(2)}+a(p_j-p_i)<-(\delta_j-\delta_i)\}\\+1\{(x_j^{(2)}-x_i^{(2)})'b^{(2)}+a(p_j-p_i)>-(\delta_j-\delta_i)\}=1,~a.s.
\end{multline*}
Therefore, $\sum_{\mathcal J\subseteq\{1,\cdots,J\}}\phi_j\circ\mathcal M_{\mathcal J}(x^{(2)},p,\delta)=1$.
Since $\tilde \Phi_j$ is constructed by summing $\phi_j\circ \mathcal M_{\mathcal J}$ over subsets of $\{1,\cdots,J\}$ except $\mathcal \{j\}$, we are left with the integral of the single indicator function $1\{x_j^{(2)}{}'b^{(2)}+ap_j<-\delta_j\}$ with respect to $f_\theta$. This ensures \eqref{eq:MJsum}. 

Now  define $\Phi$ as in \eqref{eq:Phi_pure}. Then, it follows that
\begin{align*}
\Phi(w,u )&=-\int 1\{w'\theta<-u \}f_{\theta
}(b^{(2)},a)d\theta \\
&=-\int_{-\infty }^{-u}\int_{P_{w,r}}f_{\theta }(b^{(2)},a)d\mu _{{%
w,r}}(b^{(2)},a)dr=-\int_{-\infty }^{-u}\mathcal{R}[f_{\theta
}](w,r)dr~.
\end{align*}%
Taking a derivative with respect to $u$ then yields
\begin{align}
	\frac{\partial \Phi(w,u)}{\partial u}=\mathcal{R}[f_{\theta }](w,u).
\end{align}
 Note that by Assumption \ref{as:support_pure} $\partial\Phi(w,u)/\partial u$ is either well-defined for all $(w,u)\in\mathbb H_+\times\mathbb R$ or only for $w$ in some open neighborhood. In the first case the theorem follows from the injectivity of the Radon transform. In the second case it follow from Lemma \ref{lem:limited_angle}.
\end{proof}

The following lemma is used in the proof of Theorem \ref{thm:idsub}.
\begin{lemma}\label{lem:inv1} 
Suppose the Assumptions \ref{as:index} and
Condition \ref{cond:bundles} hold and that $\phi_l$ is given as in Example 2
or Example 3 with $l \in \tilde{\mathbb{L}} = \{(0,1),(0,0)\}$. Then for all 
$(x^{(2)},p) = \big(x_1^{(2)}, x_2^{(2)}, p_1, p_2\big) \in \mathbb{R}%
^{2k}$ with $(x_1^{(2)},p_1) \neq (x_2^{(2)},p_2)$ the function $\phi:%
\mathbb{R}^{2k}\times\mathbb{R}^2\to[0,1]^2$ defined as 
\begin{align*}
\phi(x^{(2)}_1,x^{(2)}_2,p_1,p_2,d_1,&d_2)\\ &\equiv \left[\phi_{(0,0)}%
\left(x^{(2)}_1,x^{(2)}_2,p_1,p_2,d_1,d_2\right),\phi_{(0,1)}%
\left(x^{(2)}_1,x^{(2)}_2,p_1,p_2,d_1,d_2\right)\right]
\end{align*}
is invertible in $(d_1,d_2)$ on any bounded subset of $\mathbb{R}^2$. This
holds for other appropriate choices of $\tilde{\mathbb{L}}$ as well (e.g. $\tilde {\mathbb L}=\{(1,0),(1,1)\}$).
\end{lemma}

\begin{proof}[\rm \bf Proof of Lemma \ref{lem:inv1}]
We start with the observation that $\phi_{(0,0)}(x^{(2)},p,d)$ is monotonically decreasing in $d_1$ and also in $d_2$ while $\phi_{(0,1)}(x^{(2)},p,d)$ is monotonically decreasing in $d_1$ and monotonically increasing in $d_2$ by definition. 
Furthermore, the full support of $\epsilon_1$ and $\epsilon_2$ implies that $\phi_{(0,0)}$ and $\phi_{(0,1)}$ are strictly increasing or decreasing in $d_1$ and $d_2$
\begin{align*}
\frac{\partial \phi_{(0,0)}(x^{(2)},p,d)}{\partial d_1} < 0,\quad \frac{\partial \phi_{(0,0)}(x^{(2)},p,d)}{\partial d_2} < 0,\\
\frac{\partial \phi_{(0,1)}(x^{(2)},p,d)}{\partial d_1} < 0, \quad \frac{\partial \phi_{(0,1)}(x^{(2)},p,d)}{\partial d_2} > 0.
\end{align*}
Hence, the determinant of the Jacobian of $d \mapsto \phi(x^{(2)},p,d)$ as well as their principle minors are strictly negative for all $d \in \supp(D)$
\begin{align*}
\det(J_\phi)(x,d) &= \frac{\partial \phi_{(0,0)}(x^{(2)},p,d)}{\partial d_1} \frac{\partial \phi_{(0,1)}(x^{(2)},p,d)}{\partial d_2} - \frac{\partial \phi_{(0,1)}(x^{(2)},p,d)}{\partial d_1} \frac{\partial \phi_{(0,0)}(x^{(2)},p,d)}{\partial d_2}\\
&< 0.
\end{align*}
Thus, on every rectangular domain in $\R^2$ the assumptions of the Gale-Nikaido theorem are fulfilled. Since any bounded subset in $\R^2$ is contained in some rectangular domain, $\phi$ is invertible on any bounded subset of $\R^2$.
\end{proof}

\begin{proof}[\rm \bf Proof of Theorem \ref{thm:idsub}]
(a) First, let $\tilde{\mathbb L}=\{(1,0),(1,1)\}.$ By Condition \ref{cond:bundles} and Lemma \ref{lem:inv1}, Assumption \ref{as:inv1} is satisfied.
By Assumptions \ref{as:index}-\ref{as:npiv1} and Theorem \ref{thm:npiv1}, $\psi$ is identified. Further, the aggregate demand $\{\phi_l,l=(1,0),(1,1)\}$ is identified by Lemma \ref{lem:inv1}. Second, take $\tilde{\mathbb L}=\{(0,0),(0,1)\}.$ Then by the same argument, the aggregate demand $\{\phi_l,l=(0,0),(0,1)\}$ is identified as well. Hence, the entire aggregate demand vector $\phi$ is identified.

Recall that the demand for bundle (0,0) satisfies \eqref{eq:dem00}. Together with Assumption \ref{as:rc_blp} and Lebesgue's theorem the limits
\begin{align*}
\tilde\Phi_{(0,0),1}(x^{(2)}_1,p_1,\delta_1)&= -\lim_{\delta_2 \rightarrow -\infty} \phi_{(0,0)}(x^{(2)},p,\delta) \\
&= -\int 1\{x^{(2)}_1{}^{\prime}b^{(2)}+ap_1+e_1<-\delta_1\}f_{\theta}(b^{(2)},a,e,\Delta)d\theta\\
&= -\int 1\{x^{(2)}_1{}^{\prime}b^{(2)}+ap_1+e_1 <-\delta_1\}f_{\vartheta_1}(b^{(2)},a,e_1)d\vartheta_1\\
\tilde\Phi_{(0,0),2}(x^{(2)}_2,p_2,\delta_2)&= -\lim_{\delta_1 \rightarrow -\infty} \phi_{(0,0)}(x^{(2)},p,\delta) \\
&= -\int 1\{x^{(2)}_2{}^{\prime}b^{(2)}+ap_2+e_2<-\delta_2\}f_{\theta}(b^{(2)},a,e,\Delta)d\theta\\
&= -\int 1\{x^{(2)}_2{}^{\prime}b^{(2)}+ap_2+e_2 <-\delta_2\}f_{\vartheta_2}(b^{(2)},a,e_2)d\vartheta_2
\end{align*}
exist and are unique. Note that in both equations $\Delta$ and $e_1$ or $e_2$ are integrated out. Hence, the first equation connects $f_{\vartheta_1}$ to $\tilde\Phi_{(0,0),1}$ and the second equation connects $f_{\vartheta_2}$ to $\tilde\Phi_{(0,0),2}$. Following the argumentation in the proof of Theorem \ref{thm:fbeta} yields that $f_{\vartheta_1}$ and $f_{\vartheta_2}$ are identified.

As a second step we repeat the argument for $\phi_{(1,1)}$. The demand for bundle (1,1) can be written as \eqref{eq:bundle11}. By taking the limits
\begin{allowdisplaybreaks}
\begin{align*}
\tilde\Phi_{(1,1),1}(x^{(2)}_1,p_1,\delta_1) &= -\lim_{\delta_2 \rightarrow -\infty} \phi_{(1,1)}(x^{(2)},p,\delta)\\
&= - \int 1\{x^{(2)}_1{}^{\prime}b^{(2)}+ap_1+e_1 + \Delta <-\delta_1\}f_{\theta}(b^{(2)},a,e,\Delta)d\theta\\
&= -\int 1\{x^{(2)}_1{}^{\prime}b^{(2)}+ap_1+e_1 + \Delta <-\delta_1\}f_{\eta_1}(b^{(2)},a,e_1+\Delta)d\eta_1\\
\tilde\Phi_{(1,1),2}(x^{(2)}_2,p_2,\delta_2) &= -\lim_{\delta_1 \rightarrow -\infty} \phi_{(1,1)}(x^{(2)},p,\delta)\\
&= -\int 1\{x^{(2)}_2{}^{\prime}b^{(2)}+ap_2+e_2 + \Delta <-\delta_2\}f_{\theta}(b^{(2)},a,e,\Delta)d\theta\\
&= - \int 1\{x^{(2)}_2{}^{\prime}b^{(2)}+ap_2+e_2 + \Delta <-\delta_2\}f_{\vartheta_2}(b^{(2)},a,e_2 + \Delta)d\eta_2
\end{align*}
\end{allowdisplaybreaks}
and following the argument in the proof of Theorem \ref{thm:fbeta} the identification of $f_{\eta_1}$ and $f_{\eta_1}$ is proven.

(b) With $f_{\eta_j}$ for $j = 1,2$ the characteristic function $\Psi_{\Delta + \epsilon_{j}|(\beta^{(2)},\alpha)}$ of $(\Delta_{it}+\epsilon_{ijt})$ conditional on $(\beta^{(2)}_{it}, \alpha_{it})$ is identified as well. With the conditional independence assumption $\Delta_{it}\perp \epsilon_{ijt}|(\beta^{(2)}_{it},\alpha_{it})$ and $\Psi_{\epsilon_{j}|(\beta^{(2)},\alpha)}(t)\ne 0$ for almost all $t\in\mathbb R$ the densities $f_{\eta_j}$ and $f_{\vartheta_j}$ can be disentangled by the deconvolution:
\[
f_{\Delta | \beta^{(2)},\alpha} = \mathcal{F}^{-1} \left(\frac{\Psi_{\Delta + \epsilon_{j}|(\beta^{(2)},\alpha)}}{\Psi_{\epsilon_{j}|(\beta^{(2)},\alpha)}}\right),
\]
where $\mathcal{F}$ denotes the Fourier transform with respect to $\Delta$. This obviously identifies $f_{\beta^{(2)},\alpha,\Delta}$ as well. If in addition $\epsilon_{i1t}$ and $\epsilon_{i2t}$ are independent conditional on $(\beta^{(2)}_{it}, \alpha_{it})$, the density of $f_\theta$ is identified by
\[
f_\theta(b^{(2)},a,e,\Delta) = f_{\epsilon_1|\beta^{(2)},\alpha}(e_1|b^{(2)},a) \, f_{\epsilon_2|\beta^{(2)},\alpha}(e_2|\beta^{(2)},\alpha) \, f_{\beta^{(2)},\alpha,\Delta}(b^{(2)},a,\Delta)
\]
This completes the proof of the theorem.
\end{proof}

\begin{proof}[\rm \bf Proof of Theorem \ref{thm:mult1}]
First, let $\tilde{\mathbb L}=\{(2,0),(2,1)\}.$ By Condition \ref{cond:bundles} and Lemma \ref{lem:inv1}, Assumption \ref{as:inv1} is satisfied.
By Assumptions \ref{as:index}-\ref{as:npiv1} and Theorem \ref{thm:npiv1}, $\psi$ is identified. This implies that the aggregate demand $\{\phi_l,l=(2,0),(2,1)\}$ is identified. Second, take $\tilde{\mathbb L}=\{(0,0),(0,1)\}.$ Then by the same argument, the aggregate demand $\{\phi_l,l=(0,0),(0,1)\}$ is identified as well. Again by Condition \ref{cond:bundles}, we can take the limits
\begin{allowdisplaybreaks}
\begin{align*}
&\tilde\Phi_{(0,0)}(x^{(2)}_1,p_1,\delta_1) = -\lim_{\delta_2 \rightarrow -\infty} \phi_{(0,0)}(x^{(2)},p,\delta)\\
&\quad=-\int 1\{x^{(2)}_1{}^{\prime}b^{(2)}+ap_1+e_1 < -\delta_1\}f_{\theta}(b^{(2)},a,e,\Delta)d\theta\\
&\quad=-\int 1\{x^{(2)}_1{}^{\prime}b^{(2)}+ap_1+e_1<-\delta_1\}f_{(\beta^{(2)},\alpha,\epsilon_1)}(b^{(2)},a,e_1)d\theta\\
&\tilde\Phi_{(0,1)}(x^{(2)}_1,p_1,\delta_1) = -\lim_{\delta_2 \rightarrow \infty} \phi_{(0,1)}(x^{(2)},p,\delta)\\
&\quad=-\int 1\{x^{(2)}_1{}^{\prime}b^{(2)}+ap_1+e_1+\Delta_{(1,1)}>-\delta_1\}f_{\theta}(b^{(2)},a,e,\Delta_{(1,1)})d\theta\\
&\quad=-\int 1\{x^{(2)}_1{}^{\prime}b^{(2)}+ap_1+e_1+\Delta_{(1,1)}>-\delta_1\}f_{(\beta^{(2)},\alpha,\epsilon_1+\Delta_{(1,1)})}(b^{(2)},a,e+\Delta_{(1,1)})d\theta\\
&\tilde\Phi_{(2,0)}(x^{(2)}_1,p_1,\delta_1) = -\lim_{\delta_2 \rightarrow -\infty} \phi_{(2,0)}(x^{(2)},p,\delta)\\
&\quad=-\int 1\{x^{(2)}_1{}^{\prime}b^{(2)}+ap_1+e_1+\Delta_{(2,0)}>-\delta_1\}f_{\theta}(b^{(2)},a,e,\Delta_{(2,0)})d\theta\\
&\quad=-\int 1\{x^{(2)}_1{}^{\prime}b^{(2)}+ap_1+e_1+\Delta_{(2,0)}>-\delta_1\}f_{(\beta^{(2)},\alpha,\epsilon_1+\Delta_{(2,0)})}(b^{(2)},a,e_1+\Delta_{(2,0)})d\theta\\
&\tilde\Phi_{(2,1)}(x^{(2)}_1,p_1,\delta_1) = -\lim_{\delta_2 \rightarrow \infty} \phi_{(2,1)}(x^{(2)},p,\delta)\\
&\quad=\int 1\{x^{(2)}_1{}^{\prime}b^{(2)}+ap_1+e_1+\Delta_{(2,1)}-\Delta_{(1,1)}>-\delta_1\}f_{\theta}(b^{(2)},a,e,\Delta_{(1,1)})d\theta\\
&\quad=\int 1\{x^{(2)}_1{}^{\prime}b^{(2)}+ap_1+ e_1+\Delta_{(2,1)}-\Delta_{(1,1)}>-\delta_1\}\\
&\hspace{33pt}\times f_{(\beta^{(2)},\alpha,\epsilon_1+\Delta_{(2,1)}-\Delta_{(1,1)})}(b^{(2)},a,e_1+\Delta_{(2,1)}-\Delta_{(1,1)})d\theta~.
\end{align*}
\end{allowdisplaybreaks}
By the argument in the proof of Theorem \ref{thm:fbeta} and the assumption that $(X_{1t}^{(2)},P_{1t},D_{1t})$ has a full support, this identifies the joint densities of $(\beta^{(2)}_{it},\alpha_{it},\epsilon_{i1t})$, $(\beta^{(2)}_{it},\alpha_{it},\epsilon_{i1t}+\Delta_{i,(1,1),t})$, $(\beta^{(2)}_{it},\alpha_{it},\epsilon_{i1t}+\Delta_{i,(2,0),t})$, and $(\beta^{(2)}_{it},\alpha_{it},\epsilon_{i1t} + \Delta_{i,(2,1),t}-\Delta_{i,(1,1),t})$ respectively.

In what follows, the arguments are made conditional on $(\beta^{(2)}_{it},\alpha_{it})$ unless otherwise noted. By Assumption \ref{as:mult_ci} (i), we may disentangle the distribution of  $\epsilon_{i1t}$ with that of $\Delta_{i,(1,1),t}$, $\Delta_{i,(2,0),t}$, and $\Delta_{i,(2,1),t}-\Delta_{i,(1,1),t}$ respectively by deconvolution as done in the proof of Theorem \ref{thm:idsub}. Thus, the marginal densities of $\Delta_{i,(2,1),t}-\Delta_{i,(1,1),t}$ and $\Delta_{i,(1,1),t}$ are identified. Further, we note that $\Delta_{i,(2,1),t}-\Delta_{i,(1,1),t}$ is a convolution of $\Delta_{i,(2,1),t}$ and $-\Delta_{i,(1,1),t}$.
By Assumption \ref{as:mult_ci} (ii), 
Proposition 8 of Carrasco and Florens (2010) applies. Hence, the marginal density of $\Delta_{i,(2,1),t}$ is identified. By Assumption \ref{as:mult_ci} (i), $\Delta_{i,(1,1),t}\perp\Delta_{i,(2,0),t}\perp\Delta_{i,(2,1),t}$ conditional on $(\beta^{(2)}_{it},\alpha_{it},\epsilon_{i1t})$, and each of the marginal densities was identified in the previous step. Therefore, the joint density $f_{(\Delta_{(1,1)},\Delta_{(2,0)},\Delta_{(2,1)})|(\beta^{(2)}_{it},\alpha_{it},\epsilon_{i1t})}$ is identified as the product of the marginal densities. Since the density of $(\beta^{(2)}_{it},\alpha_{it},\epsilon_{i1t})$ is identified as well, we may identify the joint density $f_{\vartheta_1}$ as $f_{\vartheta_1}=f_{(\Delta_{(1,1)},\Delta_{(2,0)},\Delta_{(2,1)})|(\beta^{(2)},\alpha,\epsilon_1)}f_{(\beta^{(2)},\alpha,\epsilon_1)}$.  $f_{\vartheta_2}$ is identified as $f_{\vartheta_1}$ by Assumption \ref{as:mult_ci} (ii). By Assumption \ref{as:mult_ci} (ii) and arguing as in \eqref{eq:ftheta_id}, $f_\theta$ is identified. Given $f_\theta$, all components of $\phi$ is identified. This completes the proof of the theorem.
\end{proof}

 %\newpage

    % references for online supplement

\end{document}